 \documentclass[twocolumn]{autart}    

\usepackage{bm}
\usepackage{enumerate}
\usepackage{commath}
\usepackage{graphicx}
\graphicspath{{Pictures/}}
\usepackage{amsmath}
\usepackage{subcaption}
\usepackage{breqn}
\usepackage{cite}
\usepackage[table]{xcolor}
\usepackage{booktabs}
\usepackage{empheq}
\usepackage{amsfonts}
\usepackage{amssymb}
\let \theoremstyle\relax
\usepackage{amsthm}
\usepackage{hyperref}
\usepackage{xcolor}
\usepackage{mathrsfs}
\usepackage{accents}
\usepackage{thmtools}
\usepackage{thm-restate}
\usepackage{float}
\usepackage{comment}
\usepackage[normalem]{ulem}

\theoremstyle{plain}
\newtheorem{theorem}{Theorem}
\newtheorem{corollary}{Corollary}

\newtheorem{definition}{Definition}
\newtheorem{assumption}{Assumption}

\newtheorem{proposition}{Proposition}

\newtheorem*{problem*}{Problem}
\newtheorem*{theorem*}{Theorem}
\newtheorem{assumption*}{Assumption}

\theoremstyle{remark}
\newtheorem{remark}{Remark}

\newtheorem{example}{Example}

\newcommand{\redtext}[1]{{\color{red}#1}}
\newcommand{\bluetext}[1]{{\color{black}#1}}

\newcommand{\myvar}[1]{\bm{#1}}
\newcommand{\myvarfrak}[1]{\bm{\mathfrak{#1}}}

\newcommand{\myvardot}[1]{\dot{\myvar{#1}}}

\newcommand{\myset}[1]{\mathscr{#1}}

\begin{document}


\begin{frontmatter}

\title{
On the Undesired Equilibria Induced by Control Barrier Function Based Quadratic Programs
} 

\thanks{ This work was supported in part by the Swedish Research Council (VR), in part by the Swedish Foundation for Strategic Research (SSF), in
part by the ERC CoG LEAFHOUND, in part by the EU CANOPIES Project, and in part by the Knut and Alice Wallenberg Foundation (KAW).}

\author[]{Xiao Tan}\ead{xiaotan@kth.se},    
\author[]{Dimos V. Dimarogonas}\ead{dimos@kth.se}  

\address[]{Division of Decision and Control Systems, School of Electrical Engineering and Computer Science, KTH Royal Institute of Technology, SE-100 44, Stockholm, Sweden}  

\begin{keyword}                           

Control barrier functions; Lyapunov method; Nonlinear analysis        
\end{keyword}                             

\begin{abstract}
In this paper, we analyze the  system behavior for  general nonlinear control-affine systems when a control barrier function-induced quadratic program-based controller is employed for feedback. In particular, we characterize the existence and locations of possible equilibrium points of the closed-loop system and also provide analytical results on how design parameters affect them. Based on this analysis, a simple modification on the existing quadratic program-based controller is provided, which, without any assumptions other than those taken in the original program,  inherits the safety set forward invariance property, and further guarantees the complete elimination of undesired equilibrium points in the interior of the safety set as well as one type of boundary equilibrium points, and local asymptotic stability of the origin. Numerical examples are given alongside the theoretical discussions.
 

\end{abstract}

\end{frontmatter}


\section{Introduction}
 Dynamical system safety has \bluetext{increasingly gained} attention driven by practical needs from robotics, autonomous driving, and other safe-critical applications. One formal definition regarding system safety relates to a set of states, \bluetext{referred} to as the safety set,  that the system is supposed to evolve within. In \bluetext{the} control community, this constrained control problem has been under discussion for a long time. Two popular methods are barrier Lyapunov functions \cite{tee2009barrier} and model predictive control \cite{mayne2000constrained} from which a safe and stabilizing control law can be derived. In general, the former method suffers from the delicate design process, the sensitivity to system noise, and the unconstrained inputs. The latter method is  usually computationally heavy, thus may not be suitable for online implementation for embedded systems.

 Alternatively, the study of control barrier functions (CBFs)\cite{Xu2015a,Ames2017,Ames2019control,xiao2021high} enforces the safety set to be forward invariant and  asymptotically stable by requiring a point-wise condition on the control input. A similar point-wise condition was earlier studied  \cite{sontag1989universal}  under the concept of control Lyapunov functions (CLFs), where system stability is concerned. In  \cite{Xu2015a}, a CLF-CBF based quadratic program (CLF-CBF-QP) formulation is proposed  with an intention to provide a modular, safe, and stabilizing control design. Thanks to the increasing computational capabilities in modern control systems and its modularity design nature,  the CLF-CBF-QP formulation has been applied successfully to a wide range applications, e.g., in adaptive cruise control \cite{Xu2015a}, bipedal robot walking \cite{Hsu2015},  multi-robot coordination, verification and control \cite{glotfelter2017nonsmooth,Wang2017a,lindemann2018control}.




 However, one major limitation with   the CLF-CBF-QP formulation is that, while the controller ensures system safety, no formal guarantee has been achieved  on the system trajectories converging to the origin (the unique minimum of the CLF). This is mainly due to the relaxation on the CLF constraint in the program for the sake of its feasibility. In fact, \cite{reis2020control} shows that even for a single integrator dynamics with a circular obstacle, the program could induce non-origin equilibrium points that are locally stable. This is not desirable in a  performance-critical task\cite{xiao2021high}, where precise stabilization is also essential for the task completion. For example, in a spacecraft docking mission, while inter-collision avoidance guarantees safety, the mission would fail if the orientation of the spacecraft is not regulated precisely.
 
 There are several endeavours in the literature to achieve safe and precise stabilization  with control barrier functions.   Intuitively speaking, this is challenging because of the modular design nature, i.e., the CBF (safety) and the CLF (stability) that are designed independently could be conflicting.
 In \cite{jankovic2018robust}, local asymptotic stability is proved by a modified quadratic program assuming that the CBF constraint is inactive around the origin.  \cite{cortez2020compatibility} discusses the compatibility between the CLF and the CBF, and a sufficient  condition on the regions of attraction is proposed. The condition is however conservative and checking such conditions for general nonlinear systems remains challenging.   In our previous work \cite{xiao2021high}, by modifying a CBF candidate, the nominal control law, which can be derived from a CLF, can be implemented without any modification in an \textit{a priori} given region inside the safety set, and thus local stability follows. Yet the possible existence of  undesired equilibria is not ruled out. \cite{reis2020control} introduces an extra CBF constraint  to the original QP  which aims to remove boundary equilibria in the original QP formulation; however, the feasibility of the modified QP is only assumed. 
 
  

 In this paper we start from characterizing the existence and locations of all possible equilibrium points under a control barrier function-induced quadratic program-based controller in the closed-loop. While partial results have been reported before, here only the existence of a CLF and a CBF is assumed, removing other assumptions found in previous works. Analytical results on how the design parameter affects the equilibrium points are also discussed.  We then present a modified control barrier function-induced quadratic program, which, without any further assumptions, simultaneously guarantees the forward invariance of the safety set, the complete elimination of undesired equilibrium points in the interior of the safety set, the complete elimination of one type of boundary equilibrium points, and the local asymptotic stability of the origin. We note that the latter three properties are new compared to the previous formulation in \cite{Xu2015a,Ames2017,Ames2019control}.



\section{Preliminary}

\textit{Notation}: The operator $\nabla:C^1(\mathbb{R}^n) \to \mathbb{R}^n$ is defined as the gradient $\frac{\partial}{\partial x}$ of a scalar-valued differentiable function with respect to $\myvar{x}$. The Lie derivatives of a function $h(\myvar{x})$ for the system $\myvardot{x} = \myvarfrak{f}(\myvar{x}) + \myvarfrak{g}(\myvar{x}) \myvar{u}$ are denoted by $L_{\mathfrak{f}} h = \nabla h^\top \myvarfrak{f}(\myvar{x}) \in \mathbb{R}$ and $L_{\mathfrak{g}} h = \nabla h^\top \myvarfrak{g}(\myvar{x}) \in \mathbb{R}^{1\times m}$, respectively.  The interior and boundary of a set $\myset{A}$ are denoted $\text{Int}(\myset{A})$ and $\partial \myset{A}$, respectively. A continuous function $\alpha:[0,a) \to [0,\infty)$ for $a \in \mathbb{R}_{>0}$ is a \textit{class $\mathcal{K}$ function} if it is strictly increasing and $\alpha(0) = 0$ \cite{Khalil2002}. $\alpha:[0,\infty) \to [0,\infty)$ is called a \textit{class $\mathcal{K}_{\infty}$ function} if it is a class $\mathcal{K}$ function and $\alpha(\infty) = \infty$. 


Consider the nonlinear control affine system
\begin{equation} \label{eq:nonlinear_dyn}
    \myvardot{x} = \myvarfrak{f}(\myvar{x}) + \myvarfrak{g}(\myvar{x}) \myvar{u},
\end{equation}
where the state $\myvar{x} \in \mathbb{R}^n$, and the control input $ \myvar{u} \in \mathbb{R}^m$.   We will consider the case where $\myvarfrak{f}(\myvar{x})$ and  $\myvarfrak{g}(\myvar{x})$  are locally Lipschitz functions in $\myvar{x}$.
Denote by $ \myvar{x}(t,\myvar{x}_0) $ the solution of $ \eqref{eq:nonlinear_dyn} $ starting from $\myvar{x}(t_0) = \myvar{x}_0$. By standard ODE theory\cite{birkhoff1978ordinary}, \bluetext{if $\myvar{u}(\myvar{x})$ is locally Lipschitz, then} there exists a maximal time interval of existence $I(\myvar{x}_0 )$ and $\myvar{x}(t,\myvar{x}_0)$ is the unique solution to the differential equation \eqref{eq:nonlinear_dyn} for all $t \in I(\myvar{x}_0 ), \myvar{x}_0\in \mathbb{R}^n$. A set $\myset{A} \subset \mathbb{R}^n$ is called \textit{forward invariant}, if for any initial condition $ \myvar{x}_0 \in \myset{A}  $, $\myvar{x}(t,\myvar{x}_0)  \in \myset{A}$ for all $t \in I(\myvar{x}_0 )$. 



\begin{definition}[Extended class $\mathcal{K}$ function  \cite{Ames2017}] \label{def:extended class K}
A continuous function $\alpha:(-b,a) \to (-\infty,\infty)$ for $a,b \in \mathbb{R}_{>0}$ is an \textit{extended class $\mathcal{K}$ function} if it is strictly increasing and $\alpha(0) = 0$.
\end{definition} 
\noindent Note that the extended class $\mathcal{K}$ functions addressed in this paper will be defined for $a,b = \infty$.

\begin{definition}[CLF] \label{def:clf}
A \bluetext{smooth} positive definite function $V:\mathbb{R}^n \to \mathbb{R}$ is a control Lyapunov function (CLF) for system \eqref{eq:nonlinear_dyn} if it  satisfies:
\begin{equation} \label{eq:clf}
    \inf_{\myvar{u} \in \mathbb{R}^m} [L_{\myvarfrak{f}} V(\myvar{x}) + L_{\myvarfrak{g}}V(\myvar{x})\myvar{u}] \leq - \gamma(V(\myvar{x})), \ \forall \myvar{x}\in \mathbb{R}^n,
\end{equation}
where $\gamma: \mathbb{R}_{\ge 0} \to \mathbb{R}_{\ge 0}$ is a class $\mathcal{K}$ function.
\end{definition}

Consider the safety set $\myset{C}$ defined as a superlevel set of a \bluetext{smooth} function $h:\mathbb{R}^n \to \mathbb{R}$:
\begin{equation} \label{eq:set_c}
    \myset{C} = \{ \myvar{x}\in \mathbb{R}^n: h(\myvar{x})\ge 0 \}.
\end{equation}

\begin{definition}[CBF] \label{def:cbf}
Let set $\myset{C}$ be defined by \eqref{eq:set_c}. $h(\myvar{x})$ is a control barrier function (CBF) for system \eqref{eq:nonlinear_dyn} if there exists a locally Lipschitz extended class $\mathcal{K}$ function $\alpha$ such that:
\begin{equation} \label{eq:cbf}
    \sup_{\myvar{u} \in \mathbb{R}^m} [L_{\myvarfrak{f}}h(\myvar{x}) + L_{\myvarfrak{g}}h(\myvar{x})\myvar{u} + \alpha(h(\myvar{x}))] \ge 0, \ \forall \myvar{x}\in \mathbb{R}^n
\end{equation}
\end{definition}

In \cite{Xu2015a}, the CBF $h(\myvar{x})$ is defined over an open set $\myset{D}$ containing the safety set $\myset{C}$. Here we instead require the CBF condition to hold in $\mathbb{R}^n$ for notational simplicity. All the results in this paper remain intact even when $h(\myvar{x})$  is defined only over an open set $\myset{D}$, except that a set intersection operation with $\myset{D}$ is needed for all the sets of states in the following derivations.

Without loss of generality, \bluetext{ we say the  origin is the desired equilibrium point if it is indeed an equilibrium point of the controlled system; otherwise, there exists no desired equilibrium point. We do not assume $ \myvarfrak{f}(\myvar{0}) = \myvar{0}$.}  All the other equilibrium points are referred to as the undesired equilibrium points. We assume the following Assumption holds throughout the paper.

\vspace{2mm}
\begin{assumption} \label{ass:CLF_CBF}
The system \eqref{eq:nonlinear_dyn} is assumed to admit a CLF $V(\myvar{x})$ and a CBF $h(\myvar{x})$, \bluetext{and the origin is assumed to be} in $\textup{Int}(\myset{C})$.
\end{assumption}

\subsection{Quadratic Program Formulation}
The minimum-norm controller proposed in \cite{Xu2015a} is given by the following quadratic program with a positive scalar $p$:
\begin{align} 
    & \ \min_{(\myvar{u},\delta)\in \mathbb{R}^{m+1}} \frac{1}{2} \| \myvar{u} \|^2 + \frac{1}{2}p\delta^2 \label{eq:original_QP} \\
    s.t. \ & L_{\myvarfrak{f}} V(\myvar{x}) + L_{\myvarfrak{g}}V(\myvar{x})\myvar{u} +\gamma(V(\myvar{x})) \le \delta, \tag{CLF}\\
    & L_{\myvarfrak{f}}h(\myvar{x}) + L_{\myvarfrak{g}}h(\myvar{x})\myvar{u} + \alpha(h(\myvar{x})) \ge 0, \tag{CBF}
\end{align}
which softens the stabilization objective via the slack variable $\delta$, and thus maintains the feasibility of the QP, i.e., if $h(\myvar{x})$ is a CBF, then the quadratic program in \eqref{eq:original_QP} is always feasible. A controller $\myvar{u}(\myvar{x})$  given by the quadratic program satisfies the CBF constraint for all $\myvar{x} \in \mathbb{R}^n$. \bluetext{If $\myvar{u}(\myvar{x})$ is locally Lipschitz, then} the safety set $\myset{C}$ is forward invariant using Brezis' version of Nagumo's Theorem\cite{xiao2021high}. However, due to the relaxation in the CLF constraint, the stabilization of the system \eqref{eq:nonlinear_dyn} is generally not guaranteed.


\section{ Closed-loop system behavior} \label{sec:analysis_cls}
In this section, we investigate the point-wise solution to the quadratic program in \eqref{eq:original_QP}, the equilibrium points of the closed-loop system, and the choice of the QP parameter $p$ in \eqref{eq:original_QP}. Hereafter we denote the control input given as a solution of \eqref{eq:original_QP} as $\myvar{u}^\star(\myvar{x})$ and the closed-loop vector field $\myvarfrak{f}_{cl}(\myvar{x}) := \myvarfrak{f}(\myvar{x}) + \myvarfrak{g}(\myvar{x}) \myvar{u}^\star(\myvar{x})$. Note that here we merely assume the existence of a CLF and a CBF, thus  remove the assumptions that $\myvarfrak{g}$ is full rank as in \cite{reis2020control} or $L_{\myvarfrak{g}}h \neq \myvar{0}, \forall \myvar{x}\in \mathbb{R}^n$ as in \cite{Xu2015a,Ames2017}.


\subsection{Explicit solution to the quadratic program} \label{sec:exp_sol}

\begin{theorem} \label{thm:exact_qp_solution}
The solution to the quadratic program in \eqref{eq:original_QP} is given by 
\begin{equation} \label{eq:qp_solution}
    \myvar{u}^\star(\myvar{x}) = \left\{ \begin{array}{ll}
      \myvar{0},   & \myvar{x} \in \Omega^{\overline{clf}}_{\overline{cbf}}\cup \Omega^{\overline{clf}}_{cbf,1}, \\
       - \frac{F_h} {L_{\myvarfrak{g}}h L_{\myvarfrak{g}}h^\top}  L_{\myvarfrak{g}}h^\top,  &  \myvar{x} \in \Omega^{\overline{clf}}_{cbf,2}, \\
      - \frac{F_V}{(1/p + L_{\myvarfrak{g}}V L_{\myvarfrak{g}}V^\top)} L_{\myvarfrak{g}}V^\top,   & \myvar{x} \in \Omega^{clf}_{\overline{cbf}} \cup \Omega^{clf}_{cbf,1}, \\
         -v_1 L_{\myvarfrak{g}}V^\top  + v_2 L_{\myvarfrak{g}}h^\top,   & \myvar{x} \in \Omega^{clf}_{cbf,2},
    \end{array} \right.
\end{equation}
where $F_V(\myvar{x}): =  L_{\myvarfrak{f}} V(\myvar{x})  +\gamma(V(\myvar{x}))  $,  $F_h(\myvar{x}) :=  L_{\myvarfrak{f}} h(\myvar{x})  +\alpha(h(\myvar{x})) $, $
     \begin{bsmallmatrix}
    v_1 \\
    v_2
    \end{bsmallmatrix}: = \begin{bsmallmatrix}
    1/p + L_{\myvarfrak{g}}V L_{\myvarfrak{g}}V^\top & -L_{\myvarfrak{g}}V L_{\myvarfrak{g}}h^\top \\
    -L_{\myvarfrak{g}}V L_{\myvarfrak{g}}h^\top  & L_{\myvarfrak{g}}h L_{\myvarfrak{g}}h^\top
    \end{bsmallmatrix}^{-1} \begin{bsmallmatrix}
    F_V \\
    -F_h
    \end{bsmallmatrix}$, and the domain sets are given by
\begin{align}
    & \Omega^{\overline{clf}}_{\overline{cbf}}   = \{ \myvar{x} \in \mathbb{R}^n: F_V < 0, F_h >0 \}, \label{eq:domain_clfin_cbfin} \\
     & \Omega^{\overline{clf}}_{cbf,1}   = \{ \myvar{x} \in \mathbb{R}^n: F_V < 0, F_h = 0,  L_{\myvarfrak{g}}h = \myvar{0} \},     \label{eq:domain_clfin_cbfon_1} \\
     & \Omega^{\overline{clf}}_{cbf,2}   = \{ \myvar{x} \in \mathbb{R}^n:   F_h \le 0, \nonumber \\ 
     &  \hspace{2cm}   F_V L_{\myvarfrak{g}}h L_{\myvarfrak{g}}h^\top- F_h  L_{\myvarfrak{g}}V L_{\myvarfrak{g}}h^\top < 0\},  \label{eq:domain_clfin_cbfon} \\
& \Omega^{clf}_{\overline{cbf}} = \{ \myvar{x} \in \mathbb{R}^n: F_V \ge 0, \nonumber \\ 
& \hspace{0.7cm}  F_V L_{\myvarfrak{g}} h L_{\myvarfrak{g}} V^\top  - F_h (1/p + L_{\myvarfrak{g}}V L_{\myvarfrak{g}}V^\top)  < 0\}, \label{eq:domain_clfon_cbfin} \\
& \Omega^{clf}_{cbf,1} = \{ \myvar{x} \in \mathbb{R}^n:  F_V \ge 0, F_h = 0, L_{\myvarfrak{g}}h = 0 \}, \label{eq:domain1_clfon_cbfon} \\
& \Omega^{clf}_{cbf,2} = \{ \myvar{x} \in \mathbb{R}^n: F_V L_{\myvarfrak{g}}h L_{\myvarfrak{g}}h^\top - F_h L_{\myvarfrak{g}}V L_{\myvarfrak{g}}h^\top \ge 0, \nonumber \\ 
& \hspace{3mm} F_V L_{\myvarfrak{g}}V L_{\myvarfrak{g}}h^\top - F_h (1/p + L_{\myvarfrak{g}}V L_{\myvarfrak{g}}V^\top) \ge 0, L_{\myvarfrak{g}}h \neq \myvar{0} \}.         \label{eq:domain2_clfon_cbfon}
\end{align}

\end{theorem}
Before diving into the proof, we note that for the domain sets in \eqref{eq:qp_solution}, a bar being in place refers to the inactivity of the corresponding constraint. The subscript $\mathit{cbf},1$ refers to the case when the CBF constraint is active and $L_{\myvarfrak{g}}h = \myvar{0}$, while $\mathit{cbf},2$ refers to the case when the CBF constraint is active and $L_{\myvarfrak{g}}h \neq \myvar{0}$. 
\begin{proof}
    The Lagrangian associated to the QP \eqref{eq:original_QP} is $    \mathcal{L} = \frac{1}{2}\| \myvar{u} \|^2 + \frac{1}{2}p\delta^2 + \lambda_1(F_V + L_{\myvarfrak{g}}V \myvar{u} - \delta) -\lambda_2(F_h +L_{\myvarfrak{g}}h \myvar{u} ).$
Here $\lambda_1 \ge 0$ and $\lambda_2 \ge 0$ are the Lagrangian multipliers. The Karush-Kuhn-Tucker (KKT) conditions are
\begin{align}
    \frac{\partial\mathcal{L}}{\partial \myvar{u}} = \myvar{u} + \lambda_1 L_{\myvarfrak{g}}V^\top -\lambda_2 L_{\myvarfrak{g}}h^\top &= 0 , \label{eq:kkt_partial_u}\\
    \frac{\partial\mathcal{L}}{\partial \delta} = p\delta - \lambda_1 &=0 , \label{eq:kkt_partial_delta}\\
    \lambda_1(F_V + L_{\myvarfrak{g}}V \myvar{u}-\delta) & = 0, \label{eq:kkt_duality1}\\
    \lambda_2(F_h +L_{\myvarfrak{g}}h \myvar{u} ) & = 0. \label{eq:kkt_duality2}
\end{align}

\item[\textit{Case 1}:] Both the CLF and CBF constraints are inactive. In this case, we have
    \begin{align}
        F_V + L_{\myvarfrak{g}}V(\myvar{x})\myvar{u} &< \delta, \label{eq:kkt_clfin_cbfin_V}\\
        F_h + L_{\myvarfrak{g}}h(\myvar{x})\myvar{u} &> 0, \label{eq:kkt_clfin_cbfin_h} \\
        \lambda_1 & = 0, \label{eq:kkt_clfin_cbfin_lam1} \\
        \lambda_2 & = 0. \label{eq:kkt_clfin_cbfin_lam2} 
    \end{align}
    From \eqref{eq:kkt_partial_delta}, $\delta = \lambda_1 / p = 0$. From  \eqref{eq:kkt_partial_u} and $\lambda_1 = \lambda_2 = 0$, 
    \begin{equation} \label{eq:u_clfin_cbfin}
        \myvar{u}^\star = \myvar{0}.
    \end{equation}
    
    
    To find out the domain where this case holds, substituting \eqref{eq:u_clfin_cbfin} into \eqref{eq:kkt_clfin_cbfin_V} and \eqref{eq:kkt_clfin_cbfin_h},  and further noting that $\delta = 0$, we obtain $\Omega^{\overline{clf}}_{\overline{cbf}} $ in \eqref{eq:domain_clfin_cbfin}.

    \item[\textit{Case 2}:] The CLF constraint is inactive and the CBF constraint is active. In this case, we have
    \begin{align}
        F_V + L_{\myvarfrak{g}}V(\myvar{x})\myvar{u} &< \delta, \label{eq:kkt_clfin_cbfon_V}\\
        F_h + L_{\myvarfrak{g}}h(\myvar{x})\myvar{u} &= 0, \label{eq:kkt_clfin_cbfon_h} \\
        \lambda_1 & = 0, \label{eq:kkt_clfin_cbfon_lam1} \\
        \lambda_2 & \ge 0. \label{eq:kkt_clfin_cbfon_lam2} 
    \end{align}
    From \eqref{eq:kkt_partial_delta}, $\delta = \lambda_1 / p = 0$. We consider the following two sub-cases. 
    \begin{enumerate}
        \item[1)] $L_{\myvarfrak{g}}h = \myvar{0}$.
        Note that $\lambda_1 = 0, L_{\myvarfrak{g}}h = \myvar{0} $, then from \eqref{eq:kkt_partial_u}, 
        \begin{equation} \label{eq:u_clfin_cbfon_1}
            \myvar{u}^{\star} = \myvar{0}.
        \end{equation}
        $\lambda_2$ could be any positive scalar. To obtain the domain where this case holds, substituting  \eqref{eq:u_clfin_cbfon_1} to \eqref{eq:kkt_clfin_cbfon_V} and \eqref{eq:kkt_clfin_cbfon_h} and noting that $\delta = 0$, we obtain $\Omega^{\overline{clf}}_{cbf,1}$ in  \eqref{eq:domain_clfin_cbfon_1}.

    \item[2)] $ L_{\myvarfrak{g}}h \neq \myvar{0}$. From  \eqref{eq:kkt_partial_u} and $\lambda_1 = 0$, $ L_{\myvarfrak{g}}h \myvar{u} - \lambda_2  L_{\myvarfrak{g}}h L_{\myvarfrak{g}}h^\top = 0$. From \eqref{eq:kkt_clfin_cbfon_h}, we further obtain $\lambda_2 = -F_h /L_{\myvarfrak{g}}h L_{\myvarfrak{g}}h^\top$, and, from \eqref{eq:kkt_partial_u},
    \begin{equation} \label{eq:u_clfin_cbfon}
        \myvar{u}^\star =- \frac{F_h} {L_{\myvarfrak{g}}h L_{\myvarfrak{g}}h^\top}  L_{\myvarfrak{g}}h^\top.
    \end{equation}
    
    
    To find out the domain where this case holds, substituting \eqref{eq:u_clfin_cbfon} into \eqref{eq:kkt_clfin_cbfon_V} and noting that $\delta = 0$, we obtain that the CLF constraint being inactive implies $F_V -  \frac{F_h} {L_{\myvarfrak{g}}h L_{\myvarfrak{g}}h^\top}  L_{\myvarfrak{g}}V L_{\myvarfrak{g}}h^\top < 0 $ and  the CBF constraint being active $\lambda_2 \ge 0$ implies $ F_h \le 0 $. Thus, we obtain $ \Omega^{\overline{clf}}_{cbf,2}$ in \eqref{eq:domain_clfin_cbfon}.
    
     \end{enumerate}

    \item[\textit{Case 3}:] The CLF constraint is active and the CBF constraint is inactive. In this case, we have 
    \begin{align}
        F_V + L_{\myvarfrak{g}}V(\myvar{x})\myvar{u} &= \delta, \label{eq:kkt_clfon_cbfin_V}\\
        F_h + L_{\myvarfrak{g}}h(\myvar{x})\myvar{u} &> 0, \label{eq:kkt_clfon_cbfin_h} \\
        \lambda_1 & \ge 0, \label{eq:kkt_clfon_cbfin_lam1} \\
        \lambda_2 & = 0. \label{eq:kkt_clfon_cbfin_lam2} 
    \end{align}
    From \eqref{eq:kkt_partial_u} and \eqref{eq:kkt_clfon_cbfin_lam2}, we obtain $\myvar{u} + \lambda_1 L_{\myvarfrak{g}} V^\top = 0$, thus $L_{\myvarfrak{g}} V \myvar{u} + \lambda_1 L_{\myvarfrak{g}} V L_{\myvarfrak{g}} V^\top = 0$. Substituting $L_{\myvarfrak{g}} V \myvar{u} = - \lambda_1 L_{\myvarfrak{g}} V L_{\myvarfrak{g}} V^\top $ into \eqref{eq:kkt_clfon_cbfin_V}, we obtain $        F_V - \lambda_1 L_{\myvarfrak{g}} V L_{\myvarfrak{g}} V^\top = \delta\stackrel{\eqref{eq:kkt_partial_delta}}{=} \lambda_1/p.$
    Thus we get 
    \begin{align}
        \lambda_1 &= (p^{-1} +L_{\myvarfrak{g}} V L_{\myvarfrak{g}} V^\top )^{-1}F_V \label{eq:cbf_off_lam1} \\
        \myvar{u}^\star &= - \lambda_1 L_{\myvarfrak{g}} V^\top = - \frac{F_V}{p^{-1} +L_{\myvarfrak{g}} V L_{\myvarfrak{g}} V^\top}L_{\myvarfrak{g}} V^\top \label{eq:cbf_off_u}
    \end{align}

    
    In the domain where this case holds, $\lambda_1 \ge 0$ and $F_h + L_{\myvarfrak{g}} h \myvar{u}^\star>0$. The former implies that $F_V \geq 0$ in view of \eqref{eq:cbf_off_lam1}; 
    the latter implies $F_h - \frac{F_V}{p^{-1} +L_{\myvarfrak{g}} V L_{\myvarfrak{g}} V^\top} L_{\myvarfrak{g}} h L_{\myvarfrak{g}} V^\top  >0 $, i.e.,  $ \Omega^{clf}_{\overline{cbf}}$ in \eqref{eq:domain_clfon_cbfin}.

    \item[\textit{Case 4}:] Both the CLF constraint and the CBF constraint are active. In this case, we have
    \begin{align}
        F_V + L_{\myvarfrak{g}}V(\myvar{x})\myvar{u} &= \delta, \label{eq:kkt_clfon_cbfon_V}\\
        F_h + L_{\myvarfrak{g}}h(\myvar{x})\myvar{u} &= 0, \label{eq:kkt_clfon_cbfon_h} \\
        \lambda_1 & \ge 0,  \\
        \lambda_2 & \ge 0. 
    \end{align}
   From \eqref{eq:kkt_partial_u}, \eqref{eq:kkt_partial_delta}, we obtain $\myvar{u} = -\lambda_1  L_{\myvarfrak{g}}V^\top + \lambda_2 L_{\myvarfrak{g}}h^\top $ and $\delta = \lambda_1/p$. Substituting $\myvar{u}$ and $\delta$ into \eqref{eq:kkt_clfon_cbfon_V}
, \eqref{eq:kkt_clfon_cbfon_h}, we obtain
\begin{equation} \label{eq:kkt_clfon_cbfon_linear_eq}
    \hspace{-4mm} \begin{bmatrix}
    1/p + L_{\myvarfrak{g}}V L_{\myvarfrak{g}}V^\top & -L_{\myvarfrak{g}}V L_{\myvarfrak{g}}h^\top \\
    -L_{\myvarfrak{g}}V L_{\myvarfrak{g}}h^\top  & L_{\myvarfrak{g}}h L_{\myvarfrak{g}}h^\top
    \end{bmatrix} \begin{bmatrix}
    \lambda_1 \\
    \lambda_2
    \end{bmatrix} = \begin{bmatrix}
    F_V \\
    -F_h
    \end{bmatrix}.
\end{equation}

Denote $\Delta := \det(\begin{bsmallmatrix}
    1/p + L_{\myvarfrak{g}}V L_{\myvarfrak{g}}V^\top & -L_{\myvarfrak{g}}V L_{\myvarfrak{g}}h^\top \\
    -L_{\myvarfrak{g}}V L_{\myvarfrak{g}}h^\top  & L_{\myvarfrak{g}}h L_{\myvarfrak{g}}h^\top
    \end{bsmallmatrix})$ for brevity.   Since $ \Delta = \| L_{\myvarfrak{g}}h \|^2/p + \| L_{\myvarfrak{g}}V \|^2 \| L_{\myvarfrak{g}}h \|^2 - (L_{\myvarfrak{g}}V L_{\myvarfrak{g}}h^\top)^2 $, and  $\|\myvar{x} \|^2 \| \myvar{y} \|^2 \ge (\myvar{x}^\top \myvar{y})^2,  \forall \myvar{x,y} \in \mathbb{R}^n$, we know that $\Delta = 0$ if and only if $L_{\myvarfrak{g}}h = \myvar{0}$ for any $p>0$. We discuss the solution to \eqref{eq:kkt_clfon_cbfon_linear_eq} in the following two sub-cases.
    \begin{enumerate}
        \item[1)] $L_{\myvarfrak{g}}h = \myvar{0}$. In this case, $\Delta = 0$. From \eqref{eq:kkt_clfon_cbfon_linear_eq}, we know
        \begin{equation} \label{eq:lam1_1_clfon_cbfon}
            \lambda_1 = F_V/(1/p + L_{\myvarfrak{g}}V L_{\myvarfrak{g}}V^\top),
        \end{equation}
         $\lambda_2$ could be any positive scalar, and $F_h = 0$. 
        Furthermore, in view of \eqref{eq:kkt_partial_u}, we obtain 
        \begin{equation} \label{eq:u1_clfon_cbfon}
            \myvar{u}^\star(\myvar{x}) = - 
        \frac{F_V}{(1/p + L_{\myvarfrak{g}}V L_{\myvarfrak{g}}V^\top)} L_{\myvarfrak{g}}V^\top,
        \end{equation}
        and, in view of \eqref{eq:kkt_partial_delta}, $\delta = \frac{F_V}{(1 + p L_{\myvarfrak{g}}V L_{\myvarfrak{g}}V^\top)}$. 
        
        In this subcase, we assumed that both the CLF and the CBF constraints are active and $L_{\myvarfrak{g}}h = 0$, which implies $\lambda_1 \ge 0, \lambda_2 \ge 0$. Note that $\lambda_1 \ge 0$ is equivalent to $F_V \ge 0$ in view of \eqref{eq:lam1_1_clfon_cbfon}. In view of \eqref{eq:kkt_clfon_cbfon_linear_eq} and $L_{\myvarfrak{g}}h = 0$, we obtain $F_h = 0$. Thus the domain where this subcase holds is  $ \Omega^{clf}_{cbf,1} $ in \eqref{eq:domain1_clfon_cbfon}.

        \item[2)] $L_{\myvarfrak{g}}h \neq \myvar{0}$. In this case, $\begin{bsmallmatrix}
    1/p + L_{\myvarfrak{g}}V L_{\myvarfrak{g}}V^\top & -L_{\myvarfrak{g}}V L_{\myvarfrak{g}}h^\top \\
    -L_{\myvarfrak{g}}V L_{\myvarfrak{g}}h^\top  & L_{\myvarfrak{g}}h L_{\myvarfrak{g}}h^\top
    \end{bsmallmatrix}$ is   positive definite (since $1/p + L_{\myvarfrak{g}}V L_{\myvarfrak{g}}V^\top > 0, \Delta >0$). We calculate $\begin{bsmallmatrix}
    1/p + L_{\myvarfrak{g}}V L_{\myvarfrak{g}}V^\top & -L_{\myvarfrak{g}}V L_{\myvarfrak{g}}h^\top \\
    -L_{\myvarfrak{g}}V L_{\myvarfrak{g}}h^\top  & L_{\myvarfrak{g}}h L_{\myvarfrak{g}}h^\top
    \end{bsmallmatrix}^{-1} = \Delta^{-1} \begin{bsmallmatrix}
    L_{\myvarfrak{g}}h L_{\myvarfrak{g}}h^\top & L_{\myvarfrak{g}}h L_{\myvarfrak{g}}V^\top \\
    L_{\myvarfrak{g}}h L_{\myvarfrak{g}}V^\top   & 1/p + L_{\myvarfrak{g}}V L_{\myvarfrak{g}}V^\top
    \end{bsmallmatrix} $. Thus, $\lambda_1$ and $\lambda_2$ are given by
    \begin{align}
        & \hspace{-10mm} \lambda_1 = \Delta^{-1} (F_V L_{\myvarfrak{g}}h L_{\myvarfrak{g}}h^\top - F_h L_{\myvarfrak{g}}V L_{\myvarfrak{g}}h^\top ), \label{eq:lam1_2_clfon_cbfon} \\
       & \hspace{-10mm} \lambda_2 = \Delta^{-1} (F_V L_{\myvarfrak{g}}V L_{\myvarfrak{g}}h^\top - F_h (1/p + L_{\myvarfrak{g}}V L_{\myvarfrak{g}}V^\top) ). \label{eq:lam2_2_clfon_cbfon} 
    \end{align}
    \end{enumerate}
From \eqref{eq:kkt_partial_u}, we obtain 
\begin{equation} \label{eq:u2_clfon_cbfon}
\myvar{u}^\star  = -\lambda_1 L_{\myvarfrak{g}}V^\top  + \lambda_2 L_{\myvarfrak{g}}h^\top,
\end{equation}
with $\lambda_1$ in \eqref{eq:lam1_2_clfon_cbfon} and $ \lambda_2$ in \eqref{eq:lam2_2_clfon_cbfon}. In the domain where this case holds,  $\lambda_1 \ge 0, \lambda_2 \ge 0$, and  $L_{\myvarfrak{g}}h \neq  0$, and it implies $ \Omega^{clf}_{cbf,2}$ in \eqref{eq:domain2_clfon_cbfon}.
\end{proof}

\bluetext{
\begin{remark}[Lipschitz continuity of \eqref{eq:qp_solution}]
In above analysis, since there may exist states where $L_gh(\myvar{x})= \myvar{0}$,  the controller $\myvar{u}^\star(\myvar{x})$ is not locally Lipschitz in general. One  example is given in \cite{morris2015continuity}. Nevertheless, this does not hinder us from analyzing $\myvar{u}^{\star}$ and the equilibrium points of the closed-loop system, due to the fact that they are obtained point-wise. More discussions on Lipschitz continuity are given in Proposition \ref{prop:local lipschitz} and Remark 2.
\end{remark}
}

\subsection{Existence and locations of equilibrium points} \label{sec:exi_equ}
It is known that the quadratic program in \eqref{eq:original_QP} will induce undesired equilibria for the closed-loop system \cite{reis2020control}. Here we revisit this problem without assuming $\myvarfrak{g}$ is full rank as in \cite{reis2020control} nor $L_{\myvarfrak{g}}h \neq 0, \forall \myvar{x}\in \mathbb{R}^n$ as in \cite{Xu2015a}.

\vspace{2mm}

\begin{theorem} \label{thm:equilibrium_points}
The set of equilibrium points of the system $\dot{\myvar{x}} = \myvarfrak{f}(\myvar{x}) + \myvarfrak{g}(\myvar{x}) \myvar{u}^\star(\myvar{x})$ with the controller $\myvar{u}^{\star}$ resulting from \eqref{eq:original_QP} is given by $\myset{E} = \myset{E}^{clf}_{\overline{cbf}} \cup \myset{E}^{clf}_{cbf,1} \cup \myset{E}^{clf}_{cbf,2}$, where 
\begin{align}
    \hspace{-4mm} &  \myset{E}^{clf}_{\overline{cbf}} = \{ \myvar{x}\in \Omega^{clf}_{\overline{cbf}} \cap \text{Int}(\myset{C}): \myvarfrak{f} = p\gamma(V) \myvarfrak{g} L_{\myvarfrak{g}}V^\top  \},
    \label{eq:set_e_clfon_cbfin}\\
    \hspace{-4mm} &  \myset{E}^{clf}_{cbf,1} =  \{ \myvar{x}\in \Omega^{clf}_{cbf,1} \cap \partial \myset{C}: \myvarfrak{f} = p\gamma(V)\myvarfrak{g} L_{\myvarfrak{g}}V^\top   \}, \label{eq:set_e_clfon_cbfon1} \\
      \hspace{-4mm} &\myset{E}^{clf}_{cbf,2} =  \{ \myvar{x}\in \Omega^{clf}_{cbf,2} \cap \partial \myset{C}: \myvarfrak{f} = \lambda_1\myvarfrak{g} L_{\myvarfrak{g}}V^\top - \lambda_2 \myvarfrak{g} L_{\myvarfrak{g}}h^\top  \}, \label{eq:set_e_clfon_cbfon2}
\end{align}
with $\lambda_1 $ given in \eqref{eq:lam1_2_clfon_cbfon} and $\lambda_2$ given in \eqref{eq:lam2_2_clfon_cbfon}.
\end{theorem}

\begin{proof} 
We first show the following facts.

\noindent \textit{Fact 1}:  No equilibrium points exist when the CLF constraint in \eqref{eq:original_QP} is inactive.

 Consider the case when the CLF constraint is inactive, meaning $F_V + L_{\myvarfrak{g}}V \myvar{u} -\delta < 0 $ and $\lambda_1 = 0$. The equilibrium condition $\myvarfrak{f}_{cl}(\myvar{x}) = 0$ implies that $L_{\myvarfrak{f}_{cl}} V = 0$, thus $F_V + L_{\myvarfrak{g}}V \myvar{u} -\delta = \gamma(V) - \delta <0$. From \eqref{eq:kkt_partial_delta}, $\delta = \lambda_1 / p = 0$. Then $V(\myvar{x}) <0$, which is a contradiction since $V(\myvar{x})$ is positive definite. This is an expected conclusion since no equilibrium points at which $\dot{V}(\myvar{x})<0$ exist.  

 \noindent \textit{Fact 2}:  No equilibrium point exists in $\mathbb{R}^n \setminus \myset{C}$.
 
At an equilibrium point $\myvar{x}_{eq}$, the CBF constraint is simplified as $\alpha(h) \ge 0$, implying that the point does not lie outside of the set $\myset{C}$. This is also quite intuitive because the integral curves starting from any states outside the set $\myset{C}$ will asymptotically approach the set $\myset{C}$ so no equilibrium points exist there.

 \noindent \textit{Fact 3}:  Consider an equilibrium point $\myvar{x}_{eq}$. Then $ \myvar{x}_{eq} \in \partial \myset{C}$ if and only if the CBF constraint is active at that point.
 
 \textit{Sufficiency:} In view that the CBF constraint is active at $\myvar{x}_{eq}$, we have $L_{\myvarfrak{f}}h(\myvar{x}_{eq}) + L_{\myvarfrak{g}}h(\myvar{x}_{eq})\myvar{u} + \alpha(h(\myvar{x}_{eq})) =L_{\myvarfrak{f}_{cl}}h(\myvar{x}_{eq})+ \alpha(h(\myvar{x}_{eq})) = 0 $. Note that $\myvar{x}_{eq}$ is an equilibrium point, i.e, $\myvarfrak{f}_{cl}(\myvar{x}_{eq}) = \myvar{0}$, thus  $\alpha(h(\myvar{x}_{eq})) = 0$, which implies $\myvar{x}_{eq} \in \partial \myset{C}$.  \textit{Necessity:} Since $\myvar{x}_{eq}$ is an equilibrium point and $\myvar{x}_{eq}\in \partial \myset{C}$, i.e., $h(\myvar{x}_{eq}) = 0$, we obtain $L_{\myvarfrak{f}}h(\myvar{x}_{eq}) + L_{\myvarfrak{g}}h(\myvar{x}_{eq})\myvar{u} + \alpha(h(\myvar{x}_{eq})) =L_{\myvarfrak{f}_{cl}}h(\myvar{x}_{eq})+ \alpha(h(\myvar{x}_{eq})) = 0 $, i.e., the CBF constraint is active.  \newline
 
 From Fact 1, we know that the the equilibrium points can only exist when the CLF constraint is active, i.e., in the sets $\Omega^{clf}_{\overline{cbf}}$, $\Omega^{clf}_{cbf,1}$  and $\Omega^{clf}_{cbf,2}$. Furthermore,  the equilibrium points need to satisfy
 \begin{equation} \label{eq:closed_loop_equilibrium_points_condition}
     \myvarfrak{f}_{cl} =  \myvarfrak{f} + \myvarfrak{g}\myvar{u}^\star  = \myvar{0}.
 \end{equation}
  In the following we will discuss these three cases.

\noindent \textit{Case 1:}  Equilibrium points in  $\Omega^{clf}_{\overline{cbf}}$.     Substituting $\myvar{u}^\star(\myvar{x})$  in \eqref{eq:qp_solution} with $\myvar{x} \in \Omega^{clf}_{\overline{cbf}}$ into \eqref{eq:closed_loop_equilibrium_points_condition}, we obtain $  \myvarfrak{f}  = \frac{F_V}{p^{-1} +L_{\myvarfrak{g}} V L_{\myvarfrak{g}} V^\top} \myvarfrak{g} L_{\myvarfrak{g}} V^\top.$
In view of the facts that $\lambda_1 = \frac{F_V}{p^{-1} +L_{\myvarfrak{g}} V L_{\myvarfrak{g}} V^\top} $ in \eqref{eq:cbf_off_lam1},   $\lambda_1 = p \delta$ in \eqref{eq:kkt_partial_delta} and  $\delta =  \gamma(V)$ (as the CLF constraint is active and $\myvarfrak{f}_{cl} = \myvar{0}$), we can also characterize the equilibrium points to be $\myvarfrak{f} = p\gamma(V)\myvarfrak{g}L_{\myvarfrak{g}}V^\top$.

 
From Fact 2, we know that the equilibrium points can only be on the boundary or in the interior of the set $\myset{C}$. From Fact 3, equilibrium points lying on $\partial \myset{C}$ implies that the CBF constraint is active, thus the equilibrium points in this case lie in the interior of the set $\myset{C}$, as given in \eqref{eq:set_e_clfon_cbfin}.

 \noindent \textit{Case 2:} Equilibrium points in $\Omega^{clf}_{cbf,1}$.  Substituting $\myvar{u^\star}(\myvar{x})$ in \eqref{eq:qp_solution} into \eqref{eq:closed_loop_equilibrium_points_condition}, we obtain $ \myvarfrak{f}  = \frac{F_V}{p^{-1} +L_{\myvarfrak{g}} V L_{\myvarfrak{g}} V^\top} \myvarfrak{g} L_{\myvarfrak{g}} V^\top .$
Noting that $\lambda_1 = \frac{F_V}{p^{-1} +L_{\myvarfrak{g}} V L_{\myvarfrak{g}} V^\top} $ in \eqref{eq:lam1_1_clfon_cbfon},  $\lambda_1 = p \delta$ in \eqref{eq:kkt_partial_delta} and $ \delta = \gamma(V)$ (as the CLF constraint is active and $\myvarfrak{f}_{cl} = \myvar{0}$), we can also characterise the equilibrium points to be $\myvarfrak{f} = p \gamma(V)\myvarfrak{g} L_{\myvarfrak{g}} V^\top$. From Fact 3, we know that the equilibrium points lie on $\partial \myset{C}$, as given in \eqref{eq:set_e_clfon_cbfon1}.

\noindent \textit{Case 3:} Equilibrium points in $\Omega^{clf}_{cbf,2}$.    Substituting $\myvar{u^\star}(\myvar{x})$ in \eqref{eq:qp_solution} into \eqref{eq:closed_loop_equilibrium_points_condition}, we obtain $ \myvarfrak{f}  = \lambda_1 \myvarfrak{g} L_{\myvarfrak{g}} V^\top - \lambda_2 \myvarfrak{g} L_{\myvarfrak{g}} h^\top,$
where $\lambda_1 $ is in \eqref{eq:lam1_2_clfon_cbfon}  and   $\lambda_2$ in \eqref{eq:lam2_2_clfon_cbfon}. From Fact 3 and the CBF constraint being active, we know that the equilibrium points lie on $\partial \myset{C}$, as given in \eqref{eq:set_e_clfon_cbfon2}.\end{proof}

From Theorem \ref{thm:equilibrium_points}, we know that an equilibrium point either lies in $\textup{Int}(\myset{C})$ or $\partial \myset{C}$. We refer to these two types of equilibrium points as interior equilibria and boundary equilibria, respectively.

The following corollary is given in \cite{reis2020control}. Here we provide the proof for the sake of readability.

\vspace{2mm}
\begin{corollary} \label{col:origin}
The origin is an equilibrium point of the closed-loop system if and only if  $\myvarfrak{f}(\myvar{0}) = \myvar{0}$.
\end{corollary}
\begin{proof}
\textit{Sufficiency}: Since $\myvarfrak{f}(\myvar{0}) = \myvar{0}, V(\myvar{0}) = 0, h(\myvar{0}) >0$, then $F_V = 0, F_h = \alpha(h(\myvar{0})) >0$, thus $\myvar{0} \in \Omega^{clf}_{\overline{cbf}}$ from \eqref{eq:domain_clfon_cbfin}. Moreover, $ \myvarfrak{f} = p\gamma(V) \myvarfrak{g} L_{\myvarfrak{g}}V^\top = \myvar{0}$, from Theorem \ref{thm:equilibrium_points}, we conclude that $\myvar{0}$ is an equilibrium point of the closed-loop system. \textit{Necessity}: Since $\myvar{0}$ is an equilibrium point and $\myvar{0} \in \textup{Int}(\myset{C})$, from Theorem \ref{thm:equilibrium_points}, $\myvarfrak{f}(\myvar{0}) = p\gamma(V(\myvar{0})) \myvarfrak{g} L_{\myvarfrak{g}}V^\top = \myvar{0}$.
\end{proof}

\subsection{Choice of QP parameter}

In this subsection, we discuss the choice of different $p$'s in \eqref{eq:original_QP} and its impact on the closed-loop equilibrium points, with an intention to remove or confine undesired equilibrium points. 

{
\noindent \textit{1) Interior equilibrium points:}

We will start our discussion for equilibrium points in $\textup{Int}(\myset{C})$.} From Theorem \ref{thm:equilibrium_points}, all the equilibrium points in  $\text{Int}(\myset{C})$ are in  $\myset{E}^{clf}_{\overline{cbf}} $, where the following holds
\begin{equation} \label{eq:p_V_equation}
\myvarfrak{f} = p\gamma(V)\myvarfrak{g}L_{\myvarfrak{g}} V^\top.
\end{equation}
Note that for a given system in \eqref{eq:nonlinear_dyn}, a given CLF $V(\myvar{x})$ and a given class $\mathcal{K}$ function $\gamma(\cdot)$, $\myvarfrak{f}, \myvarfrak{g}, \gamma(V)$ and $ L_{\myvarfrak{g}} V$  are functions of the state $\myvar{x}$. We propose the following propositions on  choosing $p$.

\vspace{2mm}
\begin{proposition}
 \label{prop:eliminating_eqs}
If there exists a positive constant $p$ such that  no point in the set $ \Omega^{clf}_{\overline{cbf}} \cap \textup{Int}(\myset{C})$ except the origin satisfies \eqref{eq:p_V_equation}, then, with such a  $p$ in \eqref{eq:original_QP} applied, no equilibrium points except the origin exist in $\textup{Int}(\myset{C})$.
\end{proposition}
The proof is evident in view of Theorem \ref{thm:equilibrium_points} and Corollary \ref{col:origin} and thus omitted here. Two numerical examples are given below.

\begin{figure*}[h]
	\centering
	\begin{subfigure}[t]{0.20\linewidth}
		\includegraphics[width=\linewidth]{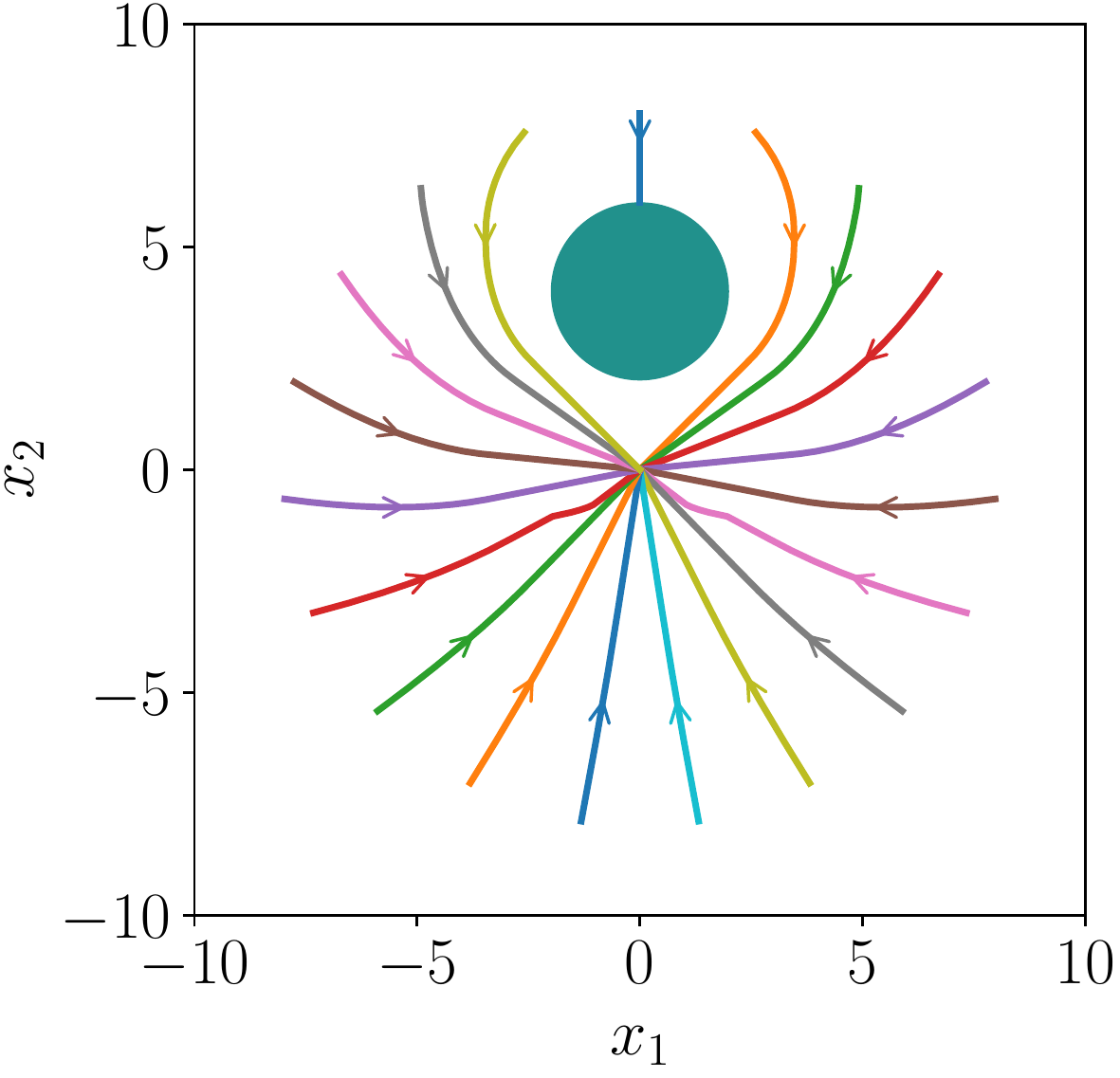}
		\caption{   $p = 1$. }   
	\end{subfigure} \hspace{5mm}
	\begin{subfigure}[t]{0.20\linewidth}
		\centering\includegraphics[width=\linewidth]{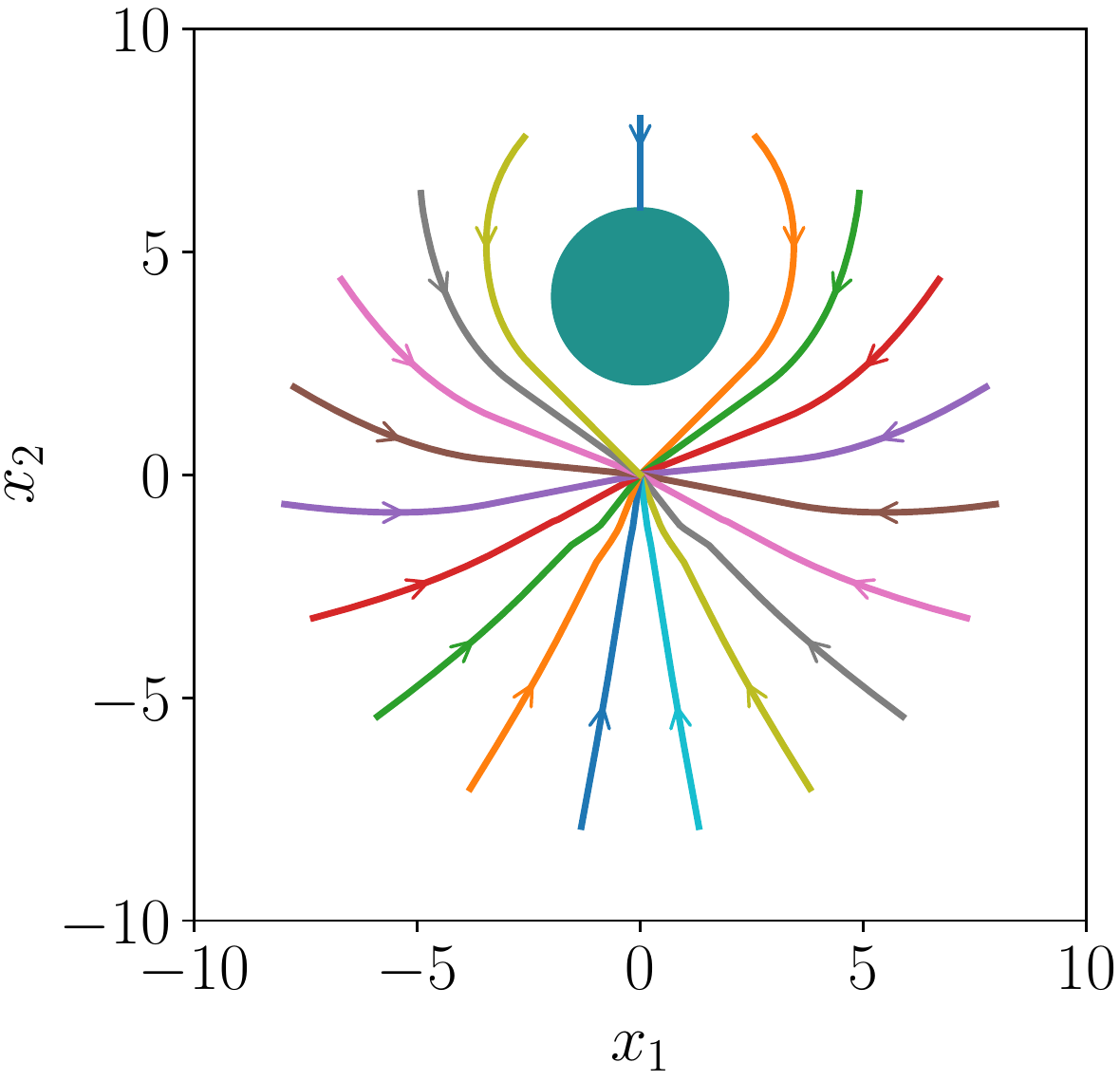}
		\caption{  $p = 10$.}
	\end{subfigure} \hspace{5mm}
	\begin{subfigure}[t]{0.20\linewidth}
		\centering\includegraphics[width=\linewidth]{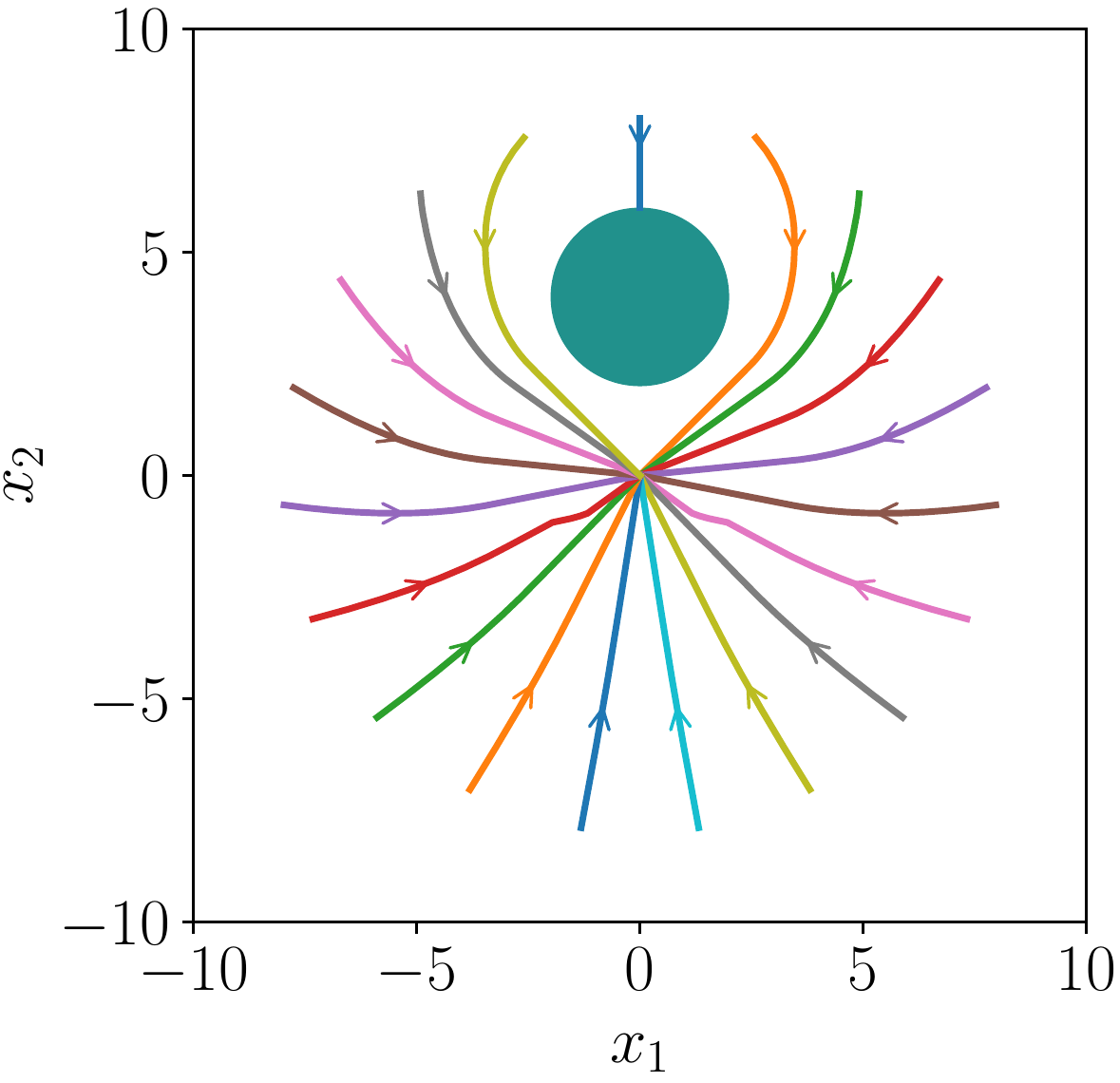}
		\caption{ $p = 100$.}
	\end{subfigure}
	\caption{ Comparison of the system trajectories in Example \ref{ex:-x} with varying $p$ values. The obstacle region is in dark green. All the simulated system trajectories converge to the origin, except one which converges to an equilibrium point on the boundary of the safety set. }  
	\label{fig:simulated_example1_trajectory}	
	\end{figure*}
	
	\begin{figure*}[h]
	\centering
	\begin{subfigure}[t]{0.20\linewidth}
		\includegraphics[width=\linewidth]{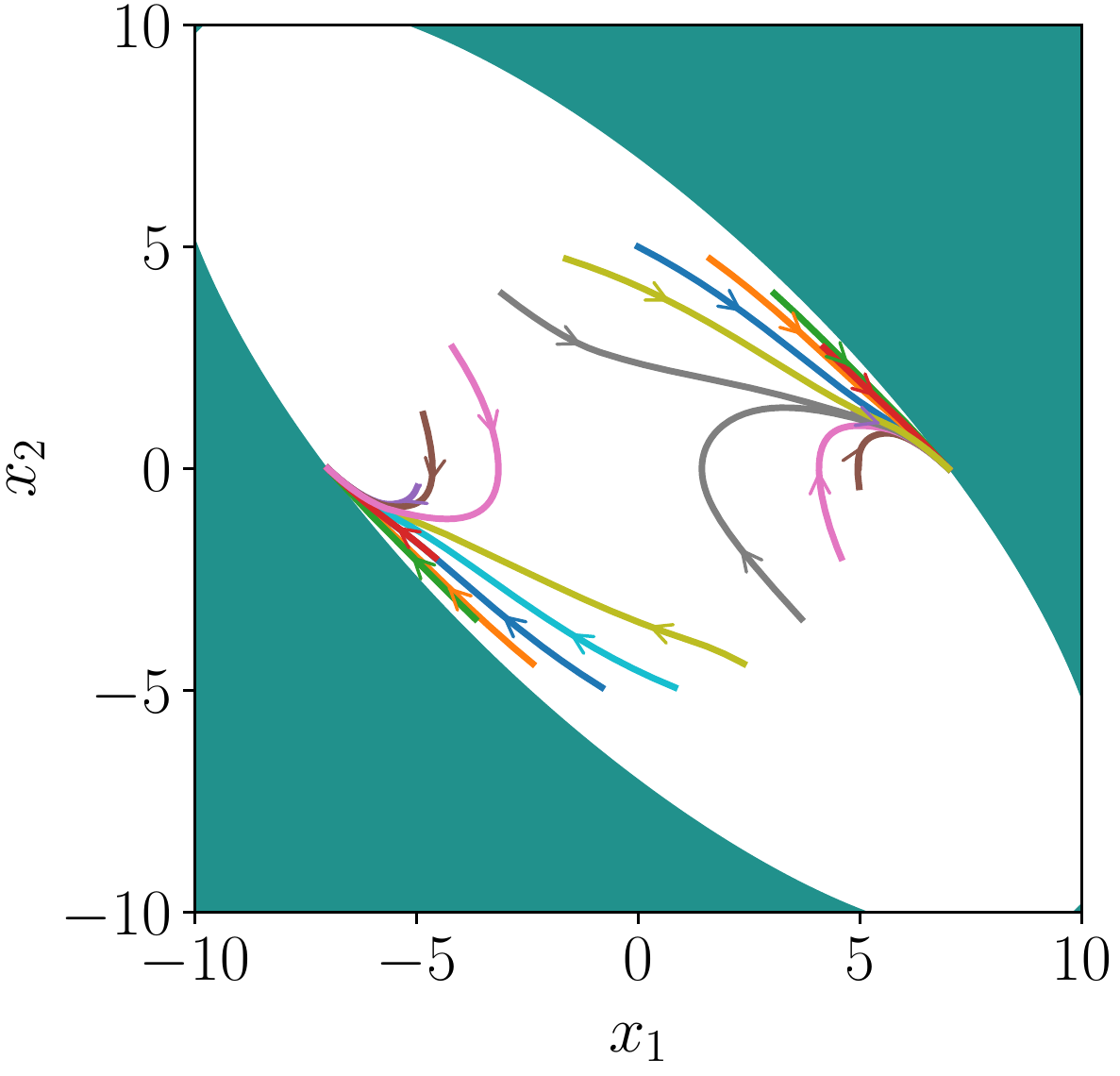}
		\caption{   $p = 0.1$. }   
	\end{subfigure} \hspace{5mm}
	\begin{subfigure}[t]{0.20\linewidth}
		\centering\includegraphics[width=\linewidth]{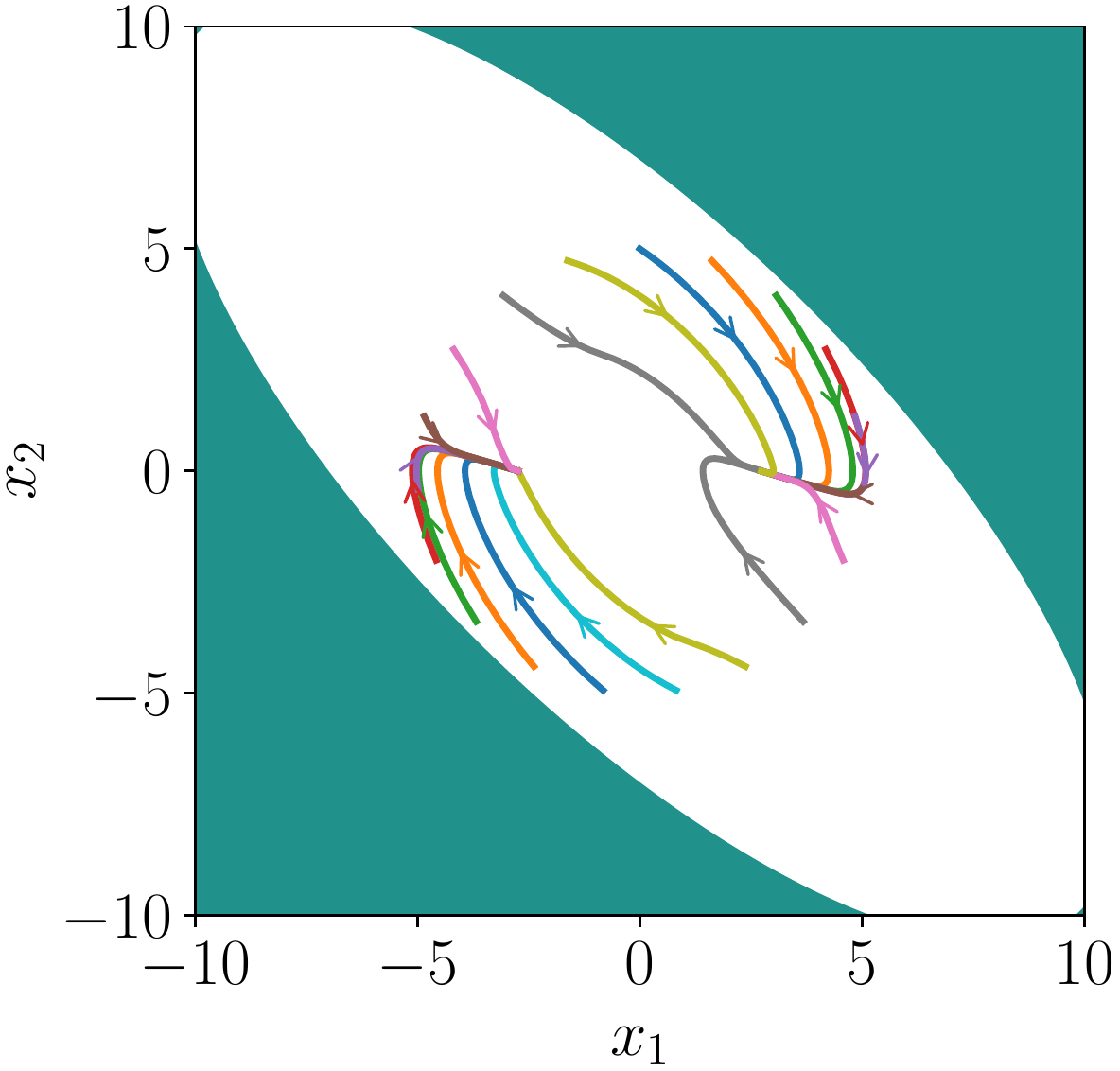}
		\caption{  $p = 1$.}
	\end{subfigure} \hspace{5mm}
	\begin{subfigure}[t]{0.20\linewidth}
		\centering\includegraphics[width=\linewidth]{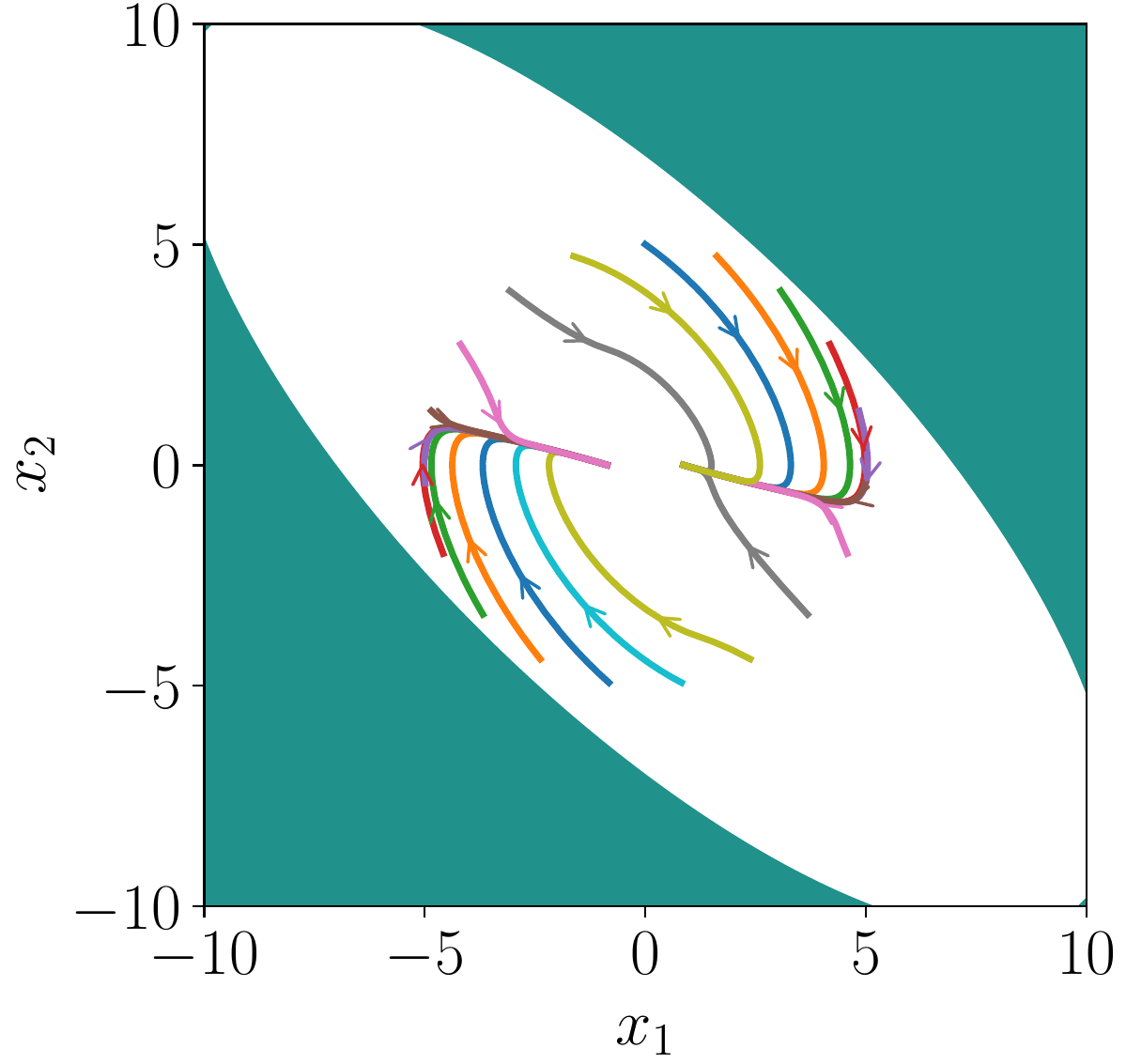}
		\caption{ $p = 10$.}
	\end{subfigure}
	\setlength{\belowcaptionskip}{-18pt}
	\caption{  Comparison of the system trajectories in Example \ref{ex:high-order} with varying $p$ values. The obstacle region is in dark green. When $p = 0.1$, the system trajectories converge to two undesired equilibrium points on the boundary of the safety set. When $p = 1$ or $10$, the system trajectories instead converge to equilibrium points in the interior of the safety set.}
	\label{fig:simulated_example2_trajectory}
	\end{figure*}
	
\vspace{2mm}
\begin{example} \label{ex:-x}
 Consider the following system 
\begin{equation} \label{eq:single_integrator_-x}
    \dot{\myvar{x}} = -\myvar{x}+ \myvar{u},
\end{equation}
 with the system state $\myvar{x} = (x_1, \ x_2)^\top$, a given CLF $V(\myvar{x}) = \frac{\myvar{x}^\top \myvar{x}}{2}$ and $\gamma(x) = x, \forall x\in \mathbb{R}_{\ge 0}$.  To show that $V(\myvar{x}) = \frac{\myvar{x}^\top \myvar{x}}{2}$ is indeed a CLF, we choose $u(\myvar{x}) =  \myvar{0}$. The time derivative of $V(\myvar{x})$
 $ \dot{V} = -x_1^2 -x_2^2 \leq - \gamma(V(\myvar{x}))$ satisfies the CLF condition in Definition \ref{def:clf}.  From \eqref{eq:p_V_equation}, by left multiplying $\nabla V^\top$ on both sides, one obtains $   L_{\myvarfrak{f}} V = p \gamma(V) L_{\myvarfrak{g}} V L_{\myvarfrak{g}} V^\top  $. Substituting $L_{\myvarfrak{f}} V = - x_1^2 -x_2^2$, $L_{\myvarfrak{g}} V = (x_1, x_2)$, $\gamma(V) = (x_1^2 + x_2^2)/2$, we obtain $-2(x_1^2 + x_2^2) =  p(x_1^2 + x_2^2)^2 $.
 Let $p $ be any positive scalar. Then this equality does not hold for any $\myvar{x} \in \mathbb{R}^2$ except the origin. Thus,  no equilibrium points except the origin exist in the interior of the set $\myset{C}$, no matter what CBF $h(\myvar{x})$ is chosen.  In  Fig. \ref{fig:simulated_example1_trajectory}, the obstacle region (in dark green) is $\{ \myvar{x} \in \mathbb{R}^2: \| \myvar{x} - (0,4)\| \leq 2  \}$ and the CBF is given by $h(\myvar{x}) =\| \myvar{x} - (0,4)\|^2 - 4$ and $\alpha(x) = x, \forall x \in \mathbb{R}$. We observe that all the simulated trajectories converge to the origin, except one that converges to an equilibrium point on the boundary of the safety set.
\end{example}

\vspace{2mm}
\begin{example} \label{ex:high-order}
 Consider the following system 
\begin{equation} \label{eq:high_order_example}
    \dot{\myvar{x}} = -\begin{psmallmatrix}
    0 & 1 \\
    1  & 0
    \end{psmallmatrix} \myvar{x}+ \begin{psmallmatrix}
    0 \\
    1  
    \end{psmallmatrix} u,
\end{equation}
 with the system state $\myvar{x} = (x_1, \ x_2)^\top$, a given CLF $V(\myvar{x}) = \frac{1}{2} x^2_1 + \frac{1}{2}(x_2 + \frac{1}{2}x_1)^2$, $\gamma(x) = \frac{3}{7} x, \forall x\in \mathbb{R}_{\ge 0}$, a given CBF $h(\myvar{x}) = -0.1x_1^2 - 0.15x_1x_2 - 0.1x_2^2+4.9$ and $\alpha(x) = x, \forall x\in \mathbb{R}$. To show that  $V(\myvar{x})$ is indeed a CLF, let $u(\myvar{x}) =  - 2x_1 - x_2$. Noticing that $\pm 2x_1 x_2 \leq x_1^2 + x_2^2$, one verifies that $V(\myvar{x}) = \frac{5}{8}x_1^2 + \frac{1}{2}x_1x_2+ \frac{1}{2}x_2^2 \leq \frac{7}{8}x_1^2 + \frac{7}{8}x_2^2$ and 
 $ \dot{V} = - \frac{1}{2}x_1^2 -\frac{1}{4}x_1 x_2 -\frac{1}{2}x_2^2 \leq - \frac{3}{8}x_1^2  - \frac{3}{8}x_2^2 \leq  - \gamma(V(\myvar{x}))$
 satisfies the CLF condition in Definition \ref{def:clf}. To show $h(\myvar{x})$ is a CBF, we only need to examine whether or not $L_{\myvarfrak{f}}h(\myvar{x}) + \alpha(h(\myvar{x})) \ge 0$ when $L_{\myvarfrak{g}}h(\myvar{x}) = -0.15x_1 - 0.2x_2 = 0$ (otherwise, with a non-zero coefficient, we can always find a $u$ that satisfies the CBF condition in Definition \ref{def:cbf}). Substituting $x_1 =(-2/1.5) x_2$ into $L_{\myvarfrak{f}}h(\myvar{x}) + \alpha(h(\myvar{x})) = -0.25x_1^2 -0.55x_1x_2 -0.25x_2^2+4.9$, one verifies that, for $\myvar{x}$ with $L_{\myvarfrak{g}}h(\myvar{x}) = 0$, $   L_{\myvarfrak{f}}h(\myvar{x}) + \alpha(h(\myvar{x})) = 0.0389x_2^2 + 4.9\ge 0.$

Suppose that there exists an equilibrium point $\myvar{x} = (x_1,x_2)\in \textup{Int}(\myset{C})$. From  \eqref{eq:p_V_equation}, $     \begin{psmallmatrix}
      x_2\\
      x_1
    \end{psmallmatrix} = p \frac{3}{7}V(\myvar{x}) (\frac{1}{2}x_1 + x_2) \begin{psmallmatrix}
    0\\
    1
    \end{psmallmatrix}.$
From the first row, we obtain $x_2 = 0$. Substituting $x_2 = 0$ into the second row, we have $x_1 = \frac{15p}{112}x_1^3$. Thus, $x_1 = 0, \pm \sqrt{112/15p}, p>0$. Proposition \ref{prop:eliminating_eqs} dictates $\myvar{x} = (x_1,x_2) \in \textup{Int}(\myset{C})$, and recall that $\myset{C}$ is the superlevel set of the CBF $h(\myvar{x})$. Thus, we conclude that for $ 0 <p< 16/105 \approx 0.152 $, there exists only one equilibrium point (the origin) in $\textup{Int}(\myset{C})$, and for $p>16/105$, there exist three equilibrium points in $\textup{Int}(\myset{C})$. This conclusion is verified by the simulation results in Fig. \ref{fig:simulated_example2_trajectory}.

Example 2 is of interest because: 1) here neither $\myvarfrak{g}$ is full rank nor $L_{\myvarfrak{g}}h \neq 0, \forall \myvar{x}\in \mathbb{R}^n$, which is required in previous works; 2) it demonstrates that, under the QP formulation in \eqref{eq:original_QP}, the existence of undesired equilibria in the interior of the safety set depends on the value of $p$.

 \end{example}

 Determining a $p$ that satisfies the assumptions in Proposition \ref{prop:eliminating_eqs} could be difficult for general nonlinear systems. One systematic way to comply with these assumptions is given in Section \ref{sec:new_qp_formulation} with a new quadratic program formulation.  Alternatively, we could tune $p$ to adjust the positions of equilibrium points in the interior of the set $\myset{C}$ as given in the following proposition.

 \vspace{2mm}
  \begin{proposition}
 \label{prop:confining_eqs}
  \bluetext{Assume that there exists a class $\mathcal{K}_{\infty}$ function $\gamma_1$ such that $\gamma_1( \| \myvar{x} \|) \leq \gamma(V(\myvar{x})) $. Let $ \bar{v}:=   \sup_{\myvar{x} \in \mathbb{R}^{n} \setminus \{ \myvar{0}\} } \frac{L_{\myvarfrak{f}}V}{L_{\myvarfrak{g}}V L_{\myvarfrak{g}}V^\top} \in [-\infty,\infty]$. If $\bar{v}$ is finite,} then all the possible equilibrium points $\myvar{x}_{eq} $ in the interior of the set $\myset{C}$ are bounded by $\| \myvar{x}_{eq} \| \leq \gamma_1^{-1}(p^{-1} \bar{v}).$

   \end{proposition}
 \begin{proof}
       \bluetext{We first show by contradiction that  at any non-origin interior equilibrium point $\myvar{x}_{eq}$, $ L_{\myvarfrak{g}}V(\myvar{x}_{eq})\neq \myvar{0}$. Suppose otherwise, then from \eqref{eq:p_V_equation}, we have $L_{\myvarfrak{f}} V (\myvar{x}_{eq}) = 0$. This however leads to a contradiction considering that $V(\myvar{x})$ is a control Lyapunov function (in \eqref{eq:clf}, the left-hand side is $0$ irrespective of $\myvar{u}$ while the right-hand side is negative). Thus, we know that for all non-origin interior equilibrium points,  $L_{\myvarfrak{g}}V(\myvar{x}_{eq})\neq \myvar{0}$, and from \eqref{eq:p_V_equation}, $\gamma(V(\myvar{x}_{eq})) = p^{-1} \frac{L_{\myvarfrak{f}}V}{L_{\myvarfrak{g}}V L_{\myvarfrak{g}}V^\top}.$ }     
Note that $\bar{v} = \sup_{\myvar{x} \in \mathbb{R}^{n} \setminus \{ \myvar{0}\} } \frac{L_{\myvarfrak{f}}V}{L_{\myvarfrak{g}}V L_{\myvarfrak{g}}V^\top} $ \bluetext{is finite by assumption}. Thus, all the possible equilibrium points in the interior of the set $\myset{C} $ are bounded by $\| \myvar{x}_{eq} \| \leq \gamma_1^{-1}(p^{-1} \bar{v}).$
 \end{proof}
 
 Proposition \ref{prop:confining_eqs} implies that we can confine the equilibrium points in the interior of the set $\myset{C}$ arbitrarily close to the origin by choosing a greater $p$. A numerical example is given below.

 \vspace{2mm}
\begin{example} \label{ex:x}
 Consider the following system 
\begin{equation} \label{eq:single_integrator_x}
    \dot{\myvar{x}} = \myvar{x}+ \myvar{u},
\end{equation}
 with the system state $\myvar{x} = (x_1, \ x_2)^\top$, a given CLF $V(\myvar{x}) = \frac{\myvar{x}^\top \myvar{x}}{2}$ and $\gamma(x) = x, \forall x\in \mathbb{R}_{\ge 0}$.
Choosing $\myvar{u}(\myvar{x}) =  (-2x_1,-2x_2)^\top$, we obtain the time derivative $ \dot{V} = -x_1^2  - x_2^2  \leq - \gamma(V(\myvar{x}))$
 satisfies the CLF condition in Definition \ref{def:clf}. 
 \bluetext{One could verify that} $\sup_{\myvar{x} \in \mathbb{R}^{2} \setminus \{ \myvar{0}\} } \frac{L_{\myvarfrak{f}}V}{L_{\myvarfrak{g}}V L_{\myvarfrak{g}}V^\top} = \sup_{\myvar{x} \in \mathbb{R}^{2} \setminus \{ \myvar{0}\} } 1 = 1 $. Thus, all possible equilibrium points $\myvar{x}_{eq}$ in the interior of the set $\myset{C}$ are bounded by $\| \myvar{x}_{eq}\| \leq \sqrt{2/p}$. In  Fig. \ref{fig:simulated_example3_trajectory}, the obstacle region (in dark green) is $\{ \myvar{x} \in \mathbb{R}^2: \| \myvar{x} - (0,4)\| \leq 2  \}$, and the CBF is given by $h(\myvar{x}) =\| \myvar{x} - (0,4)\|^2 - 4$ and $\alpha(x) = x, \forall x \in \mathbb{R}$. We observe that all of the simulated trajectories except one converge to the neighborhood region of the origin, the size of which depends on the parameter $p$. 

\end{example}

\bluetext{Similar analysis can be done for the scenario in Example 2, Fig. \ref{fig:simulated_example2_trajectory} b) and c). We omit the details for the sake of space.}

\begin{figure*}[h]
	\centering
	\begin{subfigure}[t]{0.20\linewidth}
		\includegraphics[width=\linewidth]{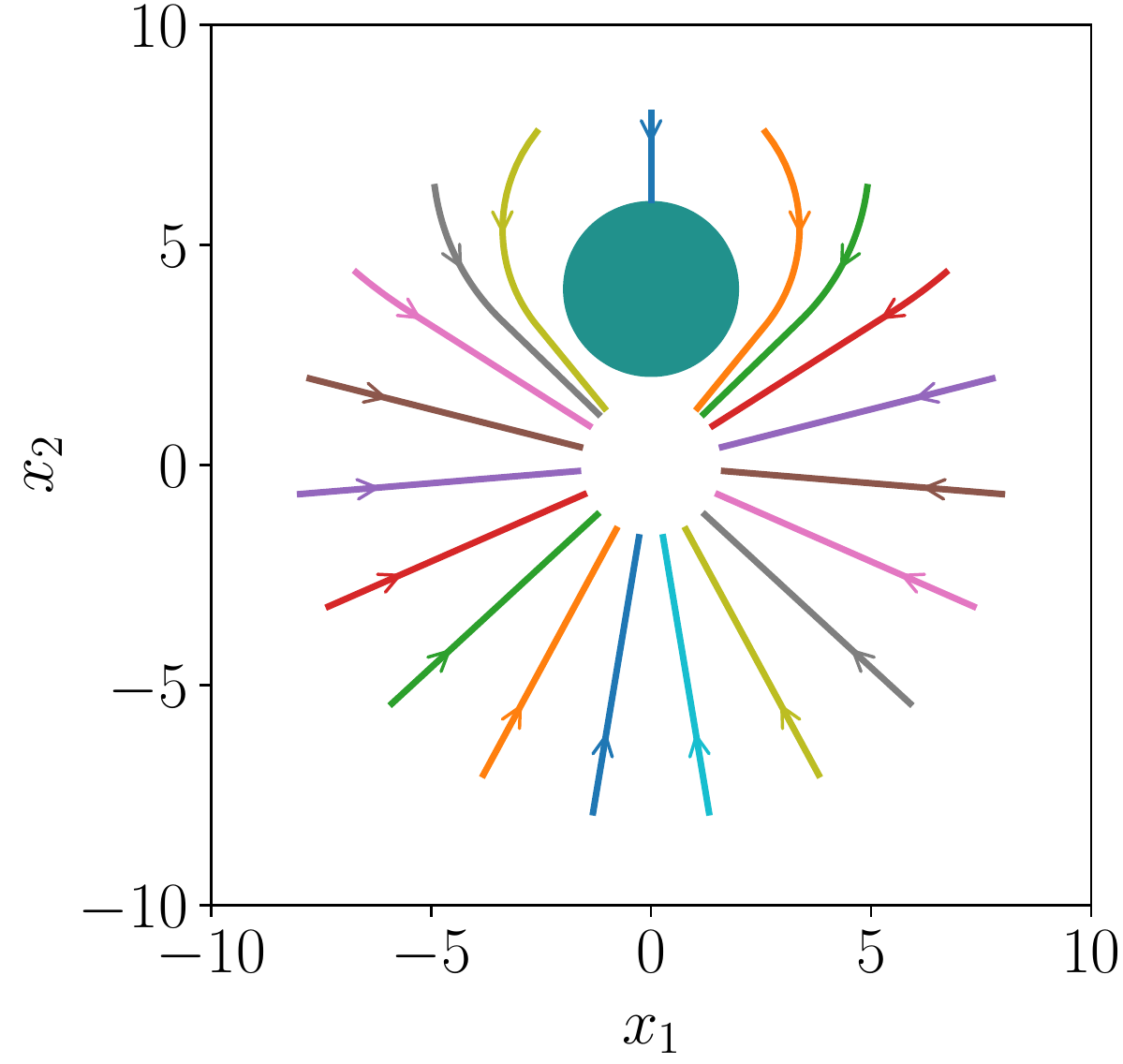}
		\caption{  $p = 1$. }   
	\end{subfigure} \hspace{5mm}
	\begin{subfigure}[t]{0.20\linewidth}
		\centering\includegraphics[width=\linewidth]{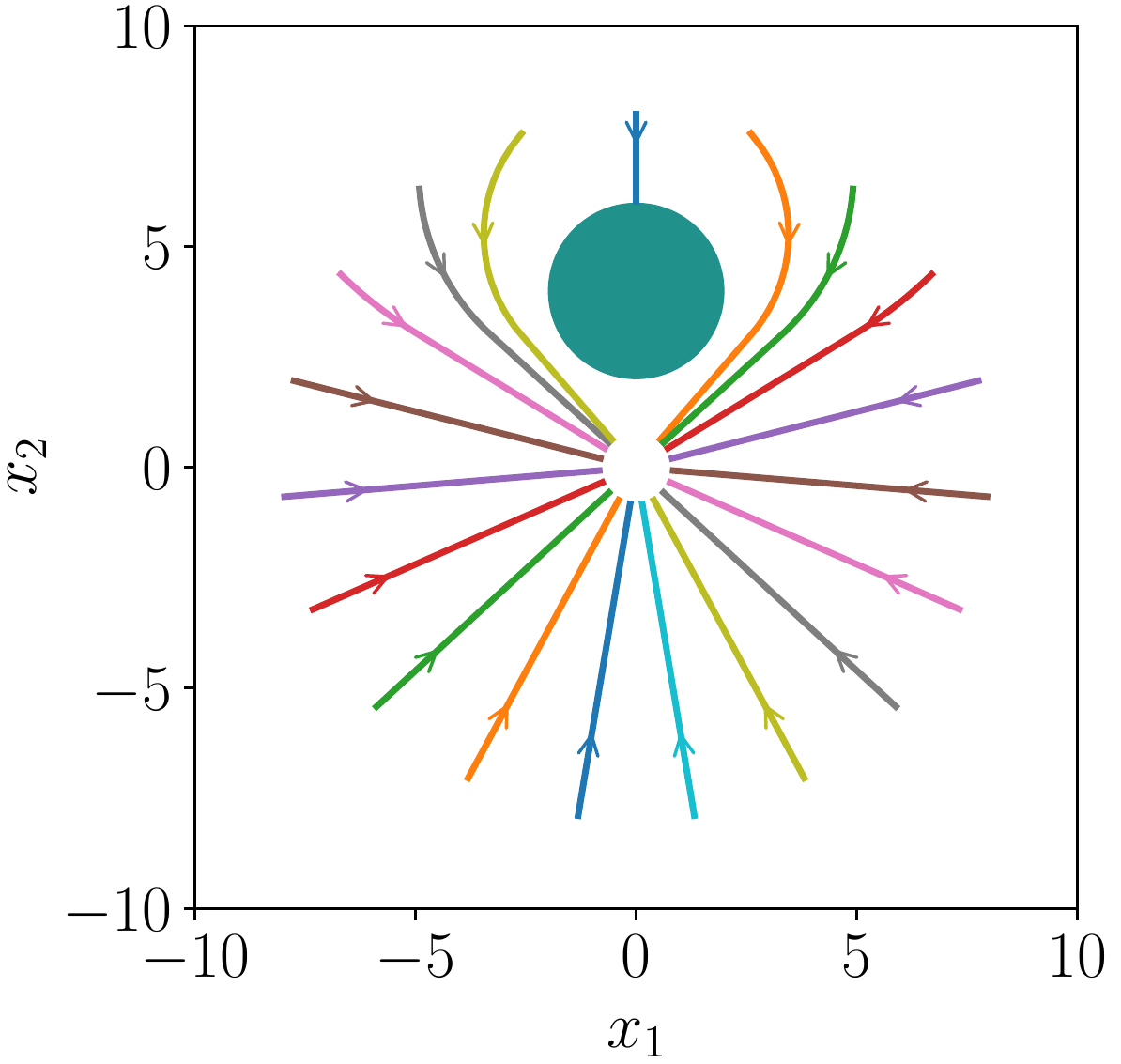}
		\caption{  $p = 10$.}
	\end{subfigure} \hspace{5mm}
	\begin{subfigure}[t]{0.20\linewidth}
		\centering\includegraphics[width=\linewidth]{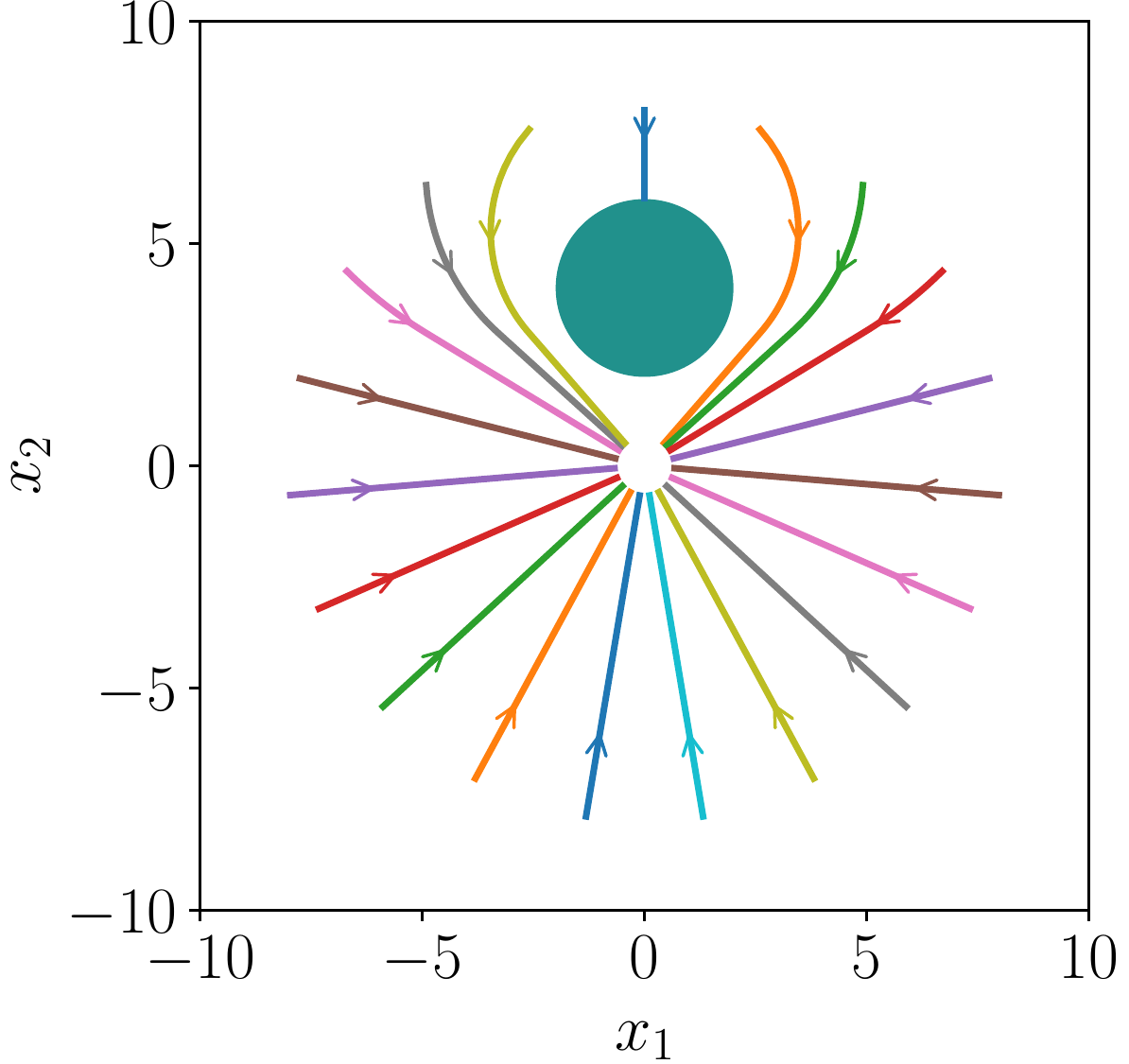}
		\caption{  $p = 100$.}
	\end{subfigure}
	\caption{ Comparison of the system trajectories in Example \ref{ex:x} with varying $p$ values.  The obstacle region is in dark green. All of the  simulated trajectories except one converge to a neighborhood region of the origin, which shrinks as $p$ becomes larger.}
	\label{fig:simulated_example3_trajectory}	
	\end{figure*}
	
	\begin{figure*}[h]
	\centering
	\begin{subfigure}[t]{0.20\linewidth}
		\includegraphics[width=\linewidth]{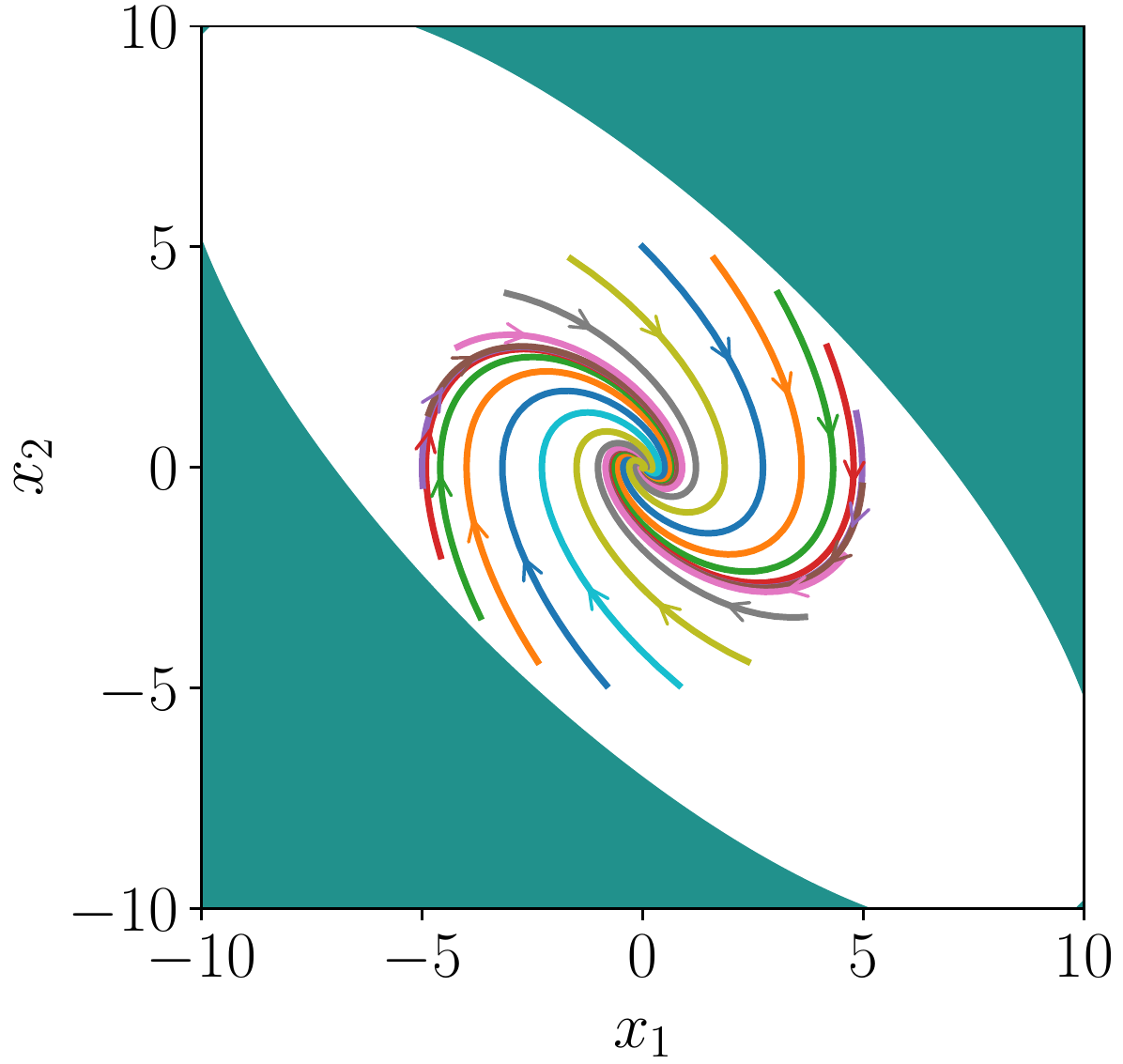}
		\caption{   $p = 0.1$. }   
	\end{subfigure} \hspace{5mm}
	\begin{subfigure}[t]{0.20\linewidth}
		\centering\includegraphics[width=\linewidth]{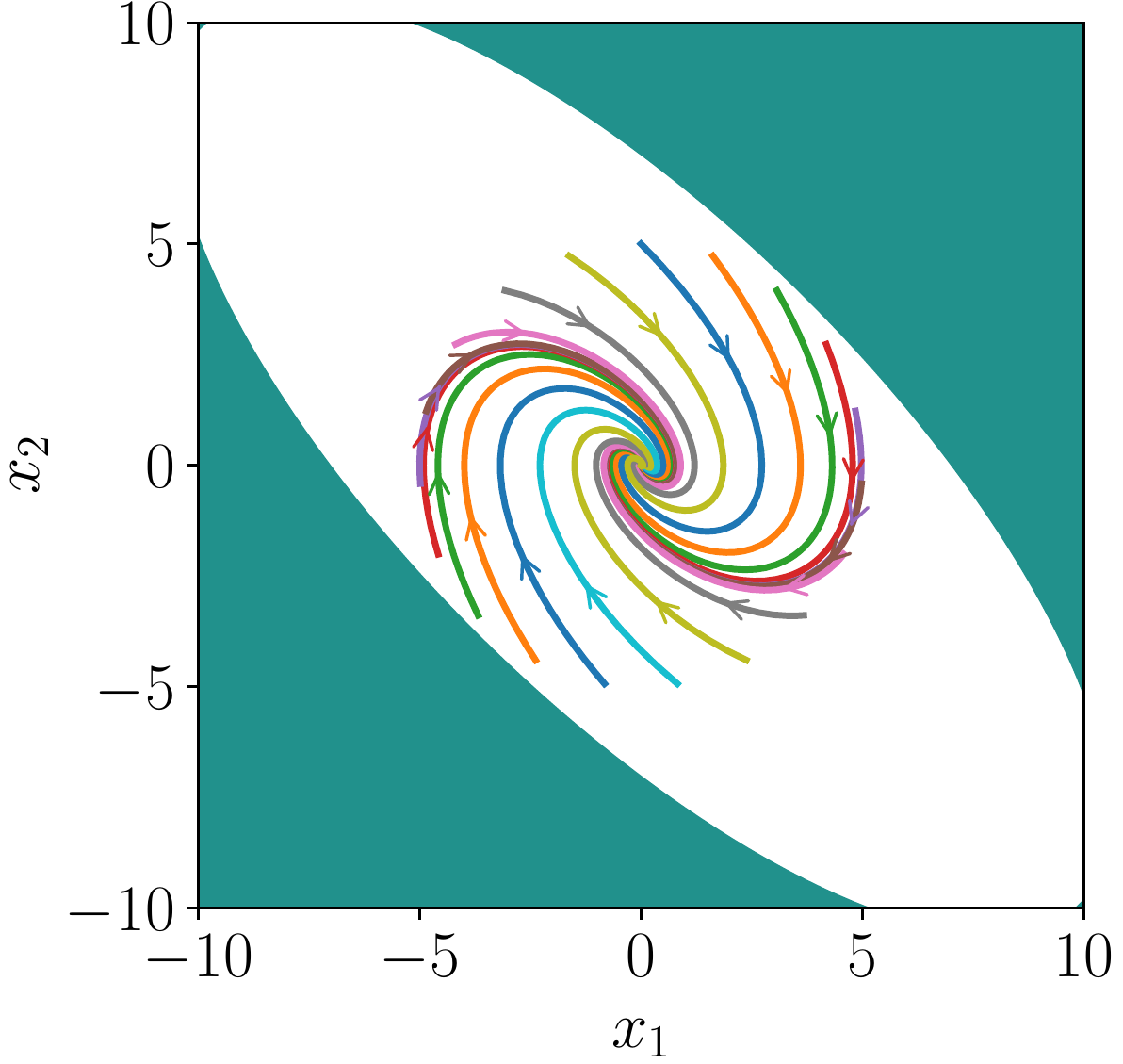}
		\caption{  $p = 1$.}
	\end{subfigure} \hspace{5mm}
	\begin{subfigure}[t]{0.20\linewidth}
		\centering\includegraphics[width=\linewidth]{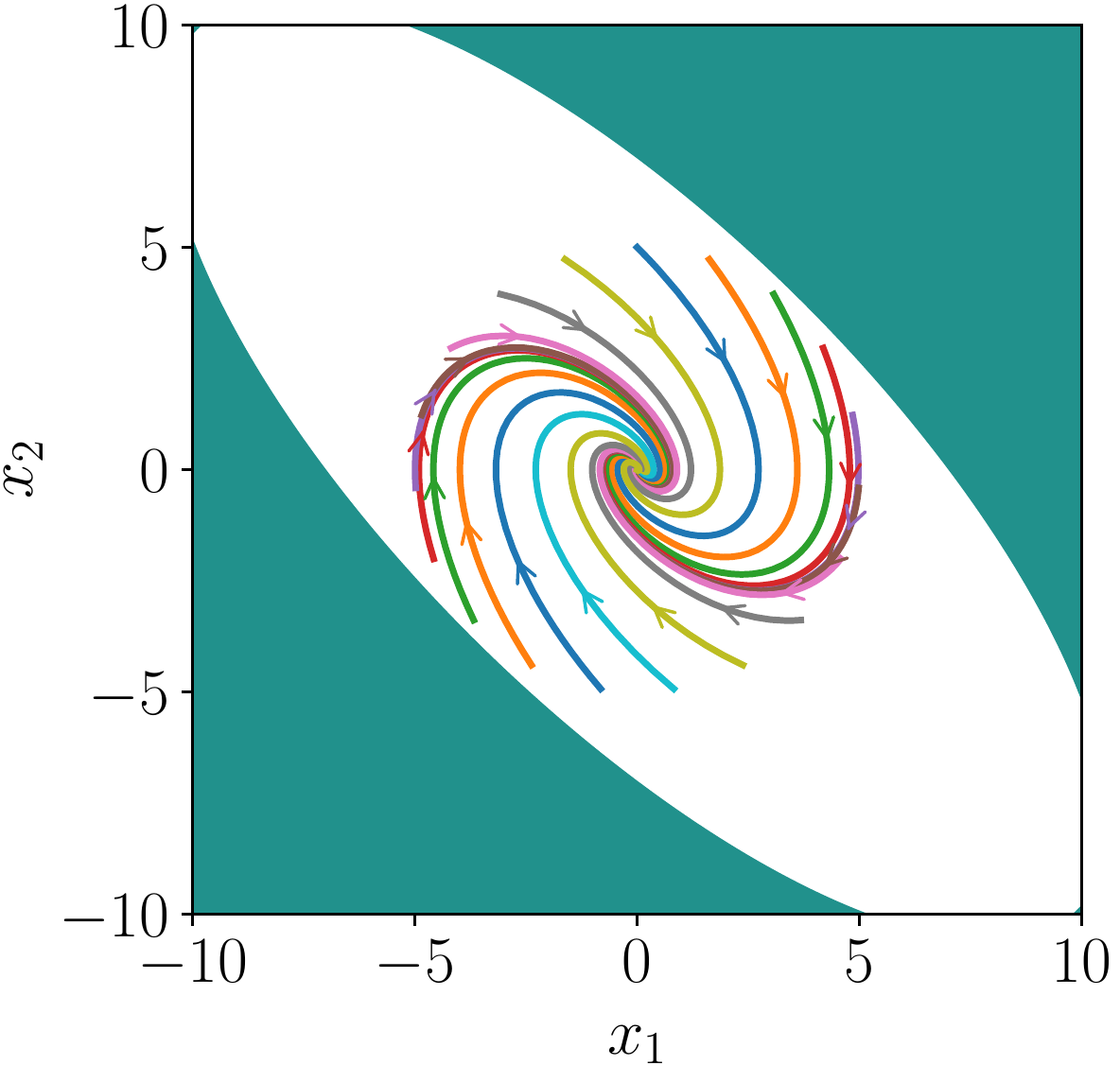}
		\caption{ $p = 10$.}
	\end{subfigure}
	\setlength{\belowcaptionskip}{-18pt}
	\caption{  Comparison of the system trajectories for the transformed system in Example \ref{ex:new_qp_high_order} with varying $p$ values. All the simulated system trajectories converge to the origin under the proposed quadratic program. }
	\label{fig:simulated_example4_trajectory}	
	\end{figure*}

{
\noindent \textit{2) Boundary equilibrium points:}}

Now consider the possible equilibrium points on $\partial \myset{C}$. For the equilibrium points in $\myset{E}^{clf}_{cbf,1}$, similar results as in Proposition \ref{prop:eliminating_eqs} and \ref{prop:confining_eqs} can be obtained as the control input shares the same form as in $\myset{E}^{clf}_{\overline{cbf}}$. For the equilibrium points in $\myset{E}^{clf}_{cbf,2}$, we show that for a particular scenario, different choices of $p$ do not affect the existence of the equilibrium points.
\vspace{2mm}
\begin{proposition} \label{prop:boundary_equilibrium}
Assume the following three conditions hold: i) $\myvar{x}_{eq} \in\myset{E}^{clf}_{cbf,2}$  for some $p>0$; ii) $\nabla V(\myvar{x}_{eq} ) = k \nabla h(\myvar{x}_{eq} )$ for some $k>0$; iii) $L_{\myvarfrak{f}}h (\myvar{x}_{eq} ) \leq 0$.
Then   $\myvar{x}_{eq}  \in\myset{E}^{clf}_{cbf,2}$ for any $ p >0$.
\end{proposition}

\begin{proof}
   From condition (i), let $p^\diamondsuit$ be the value such that $\myvar{x}_{eq} \in\myset{E}^{clf}_{cbf,2}$ when $p = p^\diamondsuit$, $p^\prime$ an arbitrary positive value, and $\lambda_1^\prime, \lambda_2^\prime$ the associated multipliers when $p = p^\prime$. To prove $\myvar{x}_{eq}  \in\myset{E}^{clf}_{cbf,2} $ for any $p>0$, by definition, we need to show that,  for $p = p^\prime$,
   \begin{align}
      \myvar{x}_{eq}\in \Omega^{clf}_{cbf,2}\cap \partial \myset{C},   \label{eq:xeq_in_Ome_Boundary} \\
        \myvarfrak{f}(\myvar{x}_{eq}) = \lambda_1^\prime \myvarfrak{g} L_{\myvarfrak{g}}V^\top(\myvar{x}_{eq}) - \lambda_2^\prime \myvarfrak{g} L_{\myvarfrak{g}}h^\top(\myvar{x}_{eq}).\label{eq:vector_field_Boundary}
   \end{align}
  This implies $\myvar{x}_{eq}  \in\myset{E}^{clf}_{cbf,2}$ for any $ p >0$, as required.  It is evident that $F_V(\myvar{x}_{eq}), F_h(\myvar{x}_{eq}), L_{\myvarfrak{g}}V(\myvar{x}_{eq}), $ $L_{\myvarfrak{g}}h(\myvar{x}_{eq})$ remain constant no matter how $p$ varies.
   
   
   \textit{Proof to \eqref{eq:xeq_in_Ome_Boundary}:} From condition (i), we know $F_V L_{\myvarfrak{g}}h L_{\myvarfrak{g}}h^\top - F_h L_{\myvarfrak{g}}V L_{\myvarfrak{g}}h^\top \ge 0,  L_{\myvarfrak{g}}h \neq \myvar{0}, \myvar{x}_{eq}\in \partial \myset{C}$. In view of definitions of $F_V, F_h$ and condition (ii), we calculate
       \begin{align}
          & F_V L_{\myvarfrak{g}}V L_{\myvarfrak{g}}h^\top - F_h (1/p^\prime + L_{\myvarfrak{g}}V L_{\myvarfrak{g}}V^\top) \nonumber \\
           & = (L_{\myvarfrak{f}}V + \gamma(V)) L_{\myvarfrak{g}}V L_{\myvarfrak{g}}h^\top - L_{\myvarfrak{f}}h(1/p^\prime + L_{\myvarfrak{g}}V L_{\myvarfrak{g}}V^\top)  \nonumber \\
           & = \gamma(V) L_{\myvarfrak{g}}V L_{\myvarfrak{g}}h^\top - 1/p^\prime L_{\myvarfrak{f}}h  \label{eq:lamda2_Boundary}
       \end{align}

      Since $\gamma(V)\ge 0,  L_{\myvarfrak{g}}V L_{\myvarfrak{g}}h^\top \ge 0$ (condition (ii)), $ 1/p^\prime\ge 0, L_{\myvarfrak{f}}h \le 0$ (condition (iii)), we get \eqref{eq:lamda2_Boundary}$\ge 0$.   Thus, $\myvar{x}_{eq}\in \Omega^{clf}_{cbf,2}\cap \partial \myset{C}$.
   

   \textit{Proof to \eqref{eq:vector_field_Boundary}:} The left-hand side (LHS) of \eqref{eq:vector_field_Boundary} is a constant, yet the right-hand side (RHS) might vary as $p^\prime$ varies. We re-write the RHS as the following function $       \myvar{s}(r) = \begin{bsmallmatrix}
       \myvar{v}_1 & \myvar{v}_2
       \end{bsmallmatrix} \begin{bsmallmatrix}
    r+a & b\\
    b  & c
    \end{bsmallmatrix}^{-1} \begin{bsmallmatrix}
    d \\
    e
    \end{bsmallmatrix}.$
   where  $r = 1/p^\prime\in (0,\infty)$; $\myvar{v}_1 :=\myvarfrak{g}L_{\myvarfrak{g}}V^\top \in \mathbb{R}^n$,  $\myvar{v}_2 := -\myvarfrak{g}L_{\myvarfrak{g}}h^\top \in \mathbb{R}^n$, $a = L_{\myvarfrak{g}}V L_{\myvarfrak{g}}V^\top, b = -L_{\myvarfrak{g}}V L_{\myvarfrak{g}}h^\top, c = L_{\myvarfrak{g}}h L_{\myvarfrak{g}}h^\top, d =  F_V,  e = -F_h $ are constants. Taking the derivative, and noting that $ \begin{bsmallmatrix}
    r+a & b\\
    b  & c
    \end{bsmallmatrix} $ is always invertible (from the proof to Theorem \ref{thm:exact_qp_solution}, in the $\Omega^{clf}_{cbf,2}$ case), we have
    \begin{equation*}
    \begin{aligned}
        \frac{d \myvar{s}(r)}{d r} & = \frac{1}{\Delta^2} \begin{bmatrix}
       \myvar{v}_1 & \myvar{v}_2
       \end{bmatrix} \begin{bmatrix}
    c & -b\\
    -b  & r+a
    \end{bmatrix} \begin{bmatrix}
    1 & 0\\
    0  & 0
    \end{bmatrix}  \begin{bmatrix}
    c & -b\\
    -b  & r+a
    \end{bmatrix} \begin{bmatrix}
    d \\
    e
    \end{bmatrix} \\
    & =  \frac{1}{\Delta^2}  (cd-be)(c\myvar{v}_1 - b\myvar{v}_2) 
    \end{aligned} 
    \end{equation*}
  Here $\Delta = \det(\begin{bsmallmatrix}
    r+a & b\\
    b  & c
    \end{bsmallmatrix})$. One verifies that $c\myvar{v}_1 - b\myvar{v}_2 =L_{\myvarfrak{g}}h L_{\myvarfrak{g}}h^\top \myvarfrak{g}L_{\myvarfrak{g}}V^\top - L_{\myvarfrak{g}}V L_{\myvarfrak{g}}h^\top \myvarfrak{g}L_{\myvarfrak{g}}h^\top =\myvar{0}$ in view of condition (ii). Thus, $\frac{d \myvar{s}(r)}{d r} = \myvar{0}$ and we obtain that the RHS of \eqref{eq:vector_field_Boundary} remains constant as $p^\prime$ varies. Note that $\myvar{s}(1/p^\diamondsuit) = \myvarfrak{f}(\myvar{x}_{eq})$, thus \eqref{eq:vector_field_Boundary} holds.
\end{proof}

\begin{example}
 Proposition \ref{prop:boundary_equilibrium} dictates that the boundary equilibrium $(0, 6)$ in Example \ref{ex:-x} will exist for any $p>0$.  This conclusion matches what we observe in Fig. \ref{fig:simulated_example1_trajectory}. 
 
 \end{example}

\section{A modified QP-based control formulation} \label{sec:new_qp_formulation}
In this section, we propose a modified CLF-CBF based control formulation.
Consider the nonlinear control affine system in \eqref{eq:nonlinear_dyn} with  a control Lyapunov function(CLF) $V$ and a control barrier function(CBF) $h$. The proposed control formulation is given as follows. Let a nominal controller $\myvar{u}_{nom}:\mathbb{R}^{n} \to \mathbb{R}^{m}$ be locally Lipschitz continuous. Rewrite \eqref{eq:nonlinear_dyn} as
\begin{equation} \label{eq:transformed_sys}
    \begin{aligned}
        \dot{\myvar{x}} =  \myvarfrak{f}^{\prime}(\myvar{x}) +\myvarfrak{g}(\myvar{x})\myvar{u}^{\prime}(\myvar{x}),
    \end{aligned}
\end{equation}
where $\myvarfrak{f}^{\prime}(\myvar{x}) := \myvarfrak{f}(\myvar{x}) +  \myvarfrak{g}(\myvar{x})\myvar{u}_{nom}(\myvar{x}),  \myvar{u}^{\prime}(\myvar{x}) := \myvar{u}(\myvar{x}) - \myvar{u}_{nom}(\myvar{x})$. In the following we will solve a new quadratic program to derive the virtual control input $\myvar{u}^{\prime}(\myvar{x})$ and the actual control input is then obtained by
\begin{equation} \label{eq:u_safe_and_stabilizing}
    \myvar{u}(\myvar{x}) = \myvar{u}_{nom}(\myvar{x}) + \myvar{u}^{\prime}(\myvar{x}).
\end{equation}
The virtual control input $\myvar{u}^\prime$ is calculated by the following quadratic program with a positive scalar $p$:
\begin{align} 
    & \  \min_{(\myvar{u}^\prime,\delta)\in \mathbb{R}^{m+1}} \frac{1}{2} \| \myvar{u}^\prime \|^2 + \frac{1}{2}p\delta^2 \label{eq:QP_with_u_prime} \\
    s.t. \ & L_{\myvarfrak{f}^\prime} V(\myvar{x}) + L_{\myvarfrak{g}}V(\myvar{x})\myvar{u}^\prime +\gamma(V(\myvar{x})) \le \delta, \tag{CLF}\\
    & L_{\myvarfrak{f}^\prime}h(\myvar{x}) + L_{\myvarfrak{g}}h(\myvar{x})\myvar{u}^\prime + \alpha(h(\myvar{x})) \ge 0. \tag{CBF}
\end{align}


\bluetext{
Before presenting our main result, we examine the continuity property of the resulting controller in \eqref{eq:u_safe_and_stabilizing}. We denote  $ F_h^{\prime}(\myvar{x})  :=  L_{\myvarfrak{f}^\prime} h(\myvar{x})  +\alpha(h(\myvar{x})), F_V^{\prime}(\myvar{x})  :=  L_{\myvarfrak{f}^\prime} V(\myvar{x})  +\gamma(V(\myvar{x})) $ in the following analysis.

\vspace{2mm}
\begin{proposition} \label{prop:local lipschitz}
 The control input $\myvar{u}: \mathbb{R}^{n} \to \mathbb{R}^{m}$  in \eqref{eq:u_safe_and_stabilizing} is locally Lipschitz continuous if one of the following conditions hold:  i) $L_{\myvarfrak{g}}h(\myvar{x}) \neq \myvar{0}$ for all $\myvar{x}$; or ii)  $\myset{M}  := \{ \myvar{x} \in \mathbb{R}^n:  F_h^{\prime}(\myvar{x}) = 0,  L_{\myvarfrak{g}}h(\myvar{x}) = \myvar{0} \}$ is empty.
\end{proposition}

\begin{proof}
    Since $\myvar{u}(\myvar{x}) = \myvar{u}_{nom}(\myvar{x}) + \myvar{u}^{\prime}(\myvar{x})$ and $\myvar{u}_{nom}(\myvar{x})$ is locally Lipschitz continuous, what we need to show is that $\myvar{u}^{\prime}(\myvar{x})$ is also locally Lipschitz continuous. Since $V(x)$ and $h(x)$ are smooth, the vector fields $\myvarfrak{f}(\myvar{x})$ and $\myvarfrak{g}(\myvar{x})$ are locally Lipschitz, thus $L_{\myvarfrak{f}^\prime} V(\myvar{x}), L_{\myvarfrak{g}}V(\myvar{x}), \gamma(V(\myvar{x})),   L_{\myvarfrak{f}^\prime}h(\myvar{x}), L_{\myvarfrak{g}}h(\myvar{x}),  \alpha(h(\myvar{x})) $ are locally Lipschitz.

    If condition (i) holds, then the coefficient matrix $\begin{bsmallmatrix}
L_{\myvarfrak{g}}V(\myvar{x}) & 1 \\
-L_{\myvarfrak{g}}h(\myvar{x}) & 0
\end{bsmallmatrix} \in \mathbb{R}^{2\times (m+1)}$ of the decision variable $(\myvar{u}^\prime, \delta)$  is of full row rank. Following \cite[Thoerem 3.1]{hager1979lipschitz}, the locally Lipschitz continuity of $\myvar{u}^\prime (\myvar{x})$ is obtained. If condition (i) does not hold, then for $\myvar{x}$ that $L_{\myvarfrak{g}}h(\myvar{x}) \neq \myvar{0}$, there exists a neighborhood $O(\myvar{x})$ in which $L_{\myvarfrak{g}}h(\myvar{y}) \neq \myvar{0}, \forall \myvar{y} \in O(\myvar{x})$, and the locally Lipschitz continuity holds following the same reasoning. In what follows we will examine the Lipschitz continuity property at which $L_{\myvarfrak{g}}h(\myvar{x}) = \myvar{0}$. 

Considering the solution to the quadratic program in \eqref{eq:QP_with_u_prime}, from Theorem \ref{thm:exact_qp_solution}, there are four potential regions when $L_{\myvarfrak{g}}h(\myvar{x}) $ could vanish: $\Omega^{\overline{clf}}_{\overline{cbf}}, \Omega^{\overline{clf}}_{cbf,1}, \Omega^{clf}_{\overline{cbf}}, \Omega^{clf}_{cbf,1}$. If condition (ii) holds, we know $\Omega^{\overline{clf}}_{cbf,1}=\Omega^{clf}_{cbf,1} = \emptyset$. Note that $\myvar{u}^{\prime}(\myvar{x}) = \myvar{0}$ for all $\myvar{x}\in \Omega^{\overline{clf}}_{\overline{cbf}}$, and $\Omega^{\overline{clf}}_{\overline{cbf}}$ is an open set, then the local Lipschitz continuity holds within $\Omega^{\overline{clf}}_{\overline{cbf}}$.       Now we check the set $ \Omega^{clf}_{\overline{cbf}}$. From the explicit form given in \eqref{eq:qp_solution},  $\myvar{u}^{\prime}(\myvar{x})$ is locally Lipschitz in $\text{Int}(\Omega^{clf}_{\overline{cbf}})$. What remains to check is the state $\myvar{x}\in \partial \Omega^{clf}_{\overline{cbf}} \cap \Omega^{clf}_{\overline{cbf}}  $ and $L_{\myvarfrak{g}}h(\myvar{x}) = \myvar{0}$, i.e., $F_V^{\prime} = 0, F_V^{\prime} L_{\myvarfrak{g}} h L_{\myvarfrak{g}} V^\top  - F_h^{\prime} (1/p + L_{\myvarfrak{g}}V L_{\myvarfrak{g}}V^\top) = - F_h^{\prime} (1/p + L_{\myvarfrak{g}}V L_{\myvarfrak{g}}V^\top)< 0$. From \eqref{eq:qp_solution}, at those points, $\myvar{u}^{\prime}(\myvar{x}) = \myvar{0}$. Take a small neighborhood $O(\myvar{x})$, then, for all $\myvar{y}\in O(\myvar{x})$, $F_h^{\prime}(\myvar{y}) >0$ (from continuity), $ F_V^{\prime}(\myvar{y}) \leq 0$ (from $\myvar{u}_{nom}$ being stabilizing), $ F_V^{\prime} L_{\myvarfrak{g}} h L_{\myvarfrak{g}} V^\top  - F_h^{\prime} (1/p + L_{\myvarfrak{g}}V L_{\myvarfrak{g}}V^\top) <0$ (from continuity).
     This implies for any $\myvar{y}\in O(\myvar{x})$,  either one of the following cases holds: a.) $F_V^{\prime}(\myvar{y}) < 0, F_h^{\prime}(\myvar{y}) >0$; b.) $F_V^{\prime}(\myvar{y}) =0, F_V^{\prime} L_{\myvarfrak{g}} h L_{\myvarfrak{g}} V^\top  - F_h^{\prime} (1/p + L_{\myvarfrak{g}}V L_{\myvarfrak{g}}V^\top)  =  - F_h^{\prime} (1/p + L_{\myvarfrak{g}}V L_{\myvarfrak{g}}V^\top) <0.$
     In both cases, from Theorem \ref{thm:exact_qp_solution}, $\myvar{u}^{\prime}(\myvar{y})=\myvar{0}$. Thus,  $\myvar{u}^{\prime}(\myvar{x})$ is locally Lipschitz continuous in this case as well. To sum up, we have shown the control input $\myvar{u}(\myvar{x})$ is locally Lipschitz continuous. 
\end{proof}

\begin{remark}
 One could verify that the first condition in  Proposition \ref{prop:local lipschitz}  fails in Example 5 and Example 6 while second condition holds. This demonstrates the generality of Proposition \ref{prop:local lipschitz}. We also note that the set $\myset{M}$ is independent of the choice of  $\myvar{u}_{nom}$ since $F_{h}^{\prime}(\myvar{x}) =  L_{\myvarfrak{f}^\prime} h(\myvar{x})  +\alpha(h(\myvar{x})) = L_{\myvarfrak{f}} h(\myvar{x})  +\alpha(h(\myvar{x})) = F_h(\myvar{x})$ when $L_{\myvarfrak{g}}h(\myvar{x}) = \myvar{0} $. 
\end{remark}
}

\vspace{2mm}
\begin{theorem} \label{thm:safe_and_stability}
Consider the nonlinear control affine system in \eqref{eq:nonlinear_dyn} with a control Lyapunov function(CLF) $V$ and a control barrier function(CBF) $h$ with its associated safety set $\myset{C}$. \bluetext{Assume one of the two conditions in Proposition \ref{prop:local lipschitz} holds.} Let the nominal control $\myvar{u}_{nom}$  satisfy the CLF condition \eqref{eq:clf}, and the control input in \eqref{eq:u_safe_and_stabilizing} is applied to \eqref{eq:nonlinear_dyn},  then 
\begin{enumerate}
    \item the set $\myset{C}$ is forward invariant;
    \item no interior equilibrium points exist except the origin;
   \item no boundary equilibrium points exist with $L_{\myvarfrak{g}}h = \myvar{0}$;
    \item  the origin is locally asymptotically stable. 
\end{enumerate}
\end{theorem}
\begin{proof}
   \bluetext{ From Proposition \ref{prop:local lipschitz}, the controller is locally Lipschitz continuous and thus the system admits a unique solution.} Consider the transformed system $\dot{\myvar{x}} = \myvarfrak{f}^{\prime}(\myvar{x}) +\myvarfrak{g}(\myvar{x})\myvar{u}^{\prime}(\myvar{x})$. Since $h$ is a CBF for the original system in \eqref{eq:nonlinear_dyn}, i.e., 
     $\forall \myvar{x} \in \mathbb{R}^{n}, \exists \myvar{u}\in \mathbb{R}^{m} \text{ such that }  L_{\myvarfrak{f}}h(\myvar{x}) + L_{\myvarfrak{g}}h(\myvar{x})\myvar{u} + \alpha(h(\myvar{x})) \ge 0 $,  we obtain that $\forall \myvar{x} \in \mathbb{R}^{n}, \myvar{u}_{nom}  \in \mathbb{R}^{m}, \exists \myvar{u}^{\prime} \in \mathbb{R}^{m} \text{ such that }  L_{\myvarfrak{f}}h(\myvar{x}) + L_{\myvarfrak{g}}h(\myvar{x})\myvar{u}_{nom} +L_{\myvarfrak{g}}h(\myvar{x})\myvar{u}^{\prime} +  \alpha(h(\myvar{x})) \ge 0 $  by choosing $\myvar{u}^{\prime} = \myvar{u} - \myvar{u}_{nom}$. Thus $h$ is also a CBF for the transformed system in \eqref{eq:transformed_sys}. It further indicates that the quadratic program in \eqref{eq:QP_with_u_prime} is feasible for all $ \myvar{x} \in \mathbb{R}^n$. $V$ is also a valid CLF for the transformed system since the CLF condition in \eqref{eq:clf} is fulfilled with $\myvar{u}^{\prime} = \myvar{0}$. Using Brezis' version of Nagumo's Theorem \cite{xiao2021high}, we further obtain that the  resulting $\myvar{u}^\prime$ will render the safety set forward invariant. 
    
    Assume that there exists an  equilibrium point $\myvar{x}_{eq}, \myvar{x}_{eq}\neq \myvar{0}$ that lies either in  $ \textup{Int}(\myset{C})$ or in $\partial \myset{C}$ with $L_{\myvarfrak{g}}h(\myvar{x}_{eq}) = \myvar{0}$. From Theorem \ref{thm:equilibrium_points}, we know that $     \myvarfrak{f}^\prime (\myvar{x}_{eq}) = p\gamma(V(\myvar{x}_{eq}))$ $ \myvarfrak{g}(\myvar{x}_{eq}) L_{\myvarfrak{g}}V^\top(\myvar{x}_{eq}).$
    By left multiplying $\nabla V^{\top}$ on both sides, we further obtain that $ L_{\myvarfrak{f}^\prime} V (\myvar{x}_{eq}) = p \gamma(V(\myvar{x}_{eq})) L_{\myvarfrak{g}} V(\myvar{x}_{eq}) L_{\myvarfrak{g}} V^\top (\myvar{x}_{eq}).$
   For any positive number $p$ and any  $\myvar{x}_{eq}\neq \myvar{0}$, on the right-hand side, we know $\gamma(V(\myvar{x}_{eq})) >0,L_{\myvarfrak{g}} V(\myvar{x}_{eq}) L_{\myvarfrak{g}} V^\top (\myvar{x}_{eq})\ge 0 $. Since $\myvar{u}_{nom}(\myvar{x})$ satisfies the CLF condition, we obtain  $L_{\myvarfrak{f}^\prime} V (\myvar{x}_{eq}) = L_{\myvarfrak{f}}V + L_{\myvarfrak{g}}V \myvar{u}_{nom} \leq - \gamma(V(\myvar{x}_{eq})) < 0 $
    on the left-hand side. Thus it yields a contradiction, implying Properties 2) and 3).

   Since $\myvar{u}_{nom}(\myvar{x})$ satisfies the CLF condition,  we have 
   $ \myvarfrak{f}^{\prime}(\myvar{0}) = \myvar{0}, 
   F^{\prime}_V := L_{\myvarfrak{f}^\prime} V  + \gamma(V) \leq 0 $ for all $\myvar{x} \in \mathbb{R}^n$. Note that $ F^{\prime}_h(\myvar{0}) = L_{\myvarfrak{f}^\prime} h(\myvar{0}) + \alpha(h(\myvar{0})) = \alpha(h(\myvar{0})) > 0$. By continuity, we know that there exists an $\epsilon >0$ such that for all $\myvar{x} \in \myset{B}_{\epsilon}:=\{ \myvar{x}\in \mathbb{R}^n: \| \myvar{x}\| \leq \epsilon\}$, $F^{\prime}_h(\myvar{x}) > 0$. Applying  Theorem \ref{thm:exact_qp_solution} with respect to the quadratic program in \eqref{eq:QP_with_u_prime}, we next show that for all $\myvar{x} \in \myset{B}_{\epsilon}$, the optimal solution is $\delta^{\star}(\myvar{x}) = 0$. This fact is obtained by examining $\delta(\myvar{x})$ in every domain and keeping in mind that 1) if $F^{\prime}_V(\myvar{x}) < 0$, then $\myvar{x} \in \Omega^{\overline{clf}}_{\overline{cbf}}$ and $\delta(\myvar{x}) = 0$; 2) if $F^{\prime}_V(\myvar{x}) =0$, then $\myvar{x}$ lies in $ \Omega^{\overline{clf}}_{cbf} \cup \Omega^{clf}_{\overline{cbf}} \cup \Omega^{clf}_{cbf,1}  $, and $\delta(\myvar{x}) = \lambda_1(\myvar{x})/p = 0$ by examining their respective $\lambda_1(\myvar{x})$s. We further obtain that $L_{\myvarfrak{f}^\prime} V(\myvar{x}) + L_{\myvarfrak{g}}V(\myvar{x})\myvar{u}^\prime +\gamma(V(\myvar{x})) \le  \delta(\myvar{x}) = 0$ for all $\myvar{x} \in \myset{B}_{\epsilon}$, i.e., $\dot{V}(\myvar{x}) = L_{\myvarfrak{f}^\prime} V(\myvar{x}) + L_{\myvarfrak{g}}V(\myvar{x})\myvar{u}^\prime \le  -\gamma(V(\myvar{x})) $ for $\myvar{x} \in \myset{B}_{\epsilon}$. With a standard Lyapunov argument\cite{Khalil2002}, we then deduce that the origin is locally asymptotically stable.    \end{proof}

\begin{remark}
In fact, any locally Lipschitz $\myvar{u}_{nom}:\mathbb{R}^n \to \mathbb{R}^m$  that renders $L_{\myvarfrak{f} } V + L_{ \myvarfrak{g} \myvar{u}_{nom}}V$ negative definite and  satisfies $\myvarfrak{f}(\myvar{0}) +\myvarfrak{g}\myvar{u}_{nom}(\myvar{0}) = \myvar{0}$ is a valid nominal controller in the new QP formulation in \eqref{eq:QP_with_u_prime}. The proof can be carried out in a similar manner.
\end{remark}

\vspace{2mm}
\begin{remark}
The proposed formulation is favorable in many regards. Assumption-wise, what it requires (Assumption 1) is the same as that of the original quadratic program \eqref{eq:original_QP}. Computation-wise, this new formulation does not add extra computations since the CLF-compatible $\myvar{u}_{nom}$ can be obtained in an analytical form\cite{sontag1989universal}. Finally, the proposed formulation provides stronger theoretical guarantees (Properties 2)-4)) on system stability while maintaining the same guarantee on system safety.
\end{remark}

\vspace{2mm}
\begin{remark}
Two CBF-based control formulations have been proposed in \cite{Xu2015a,Ames2017,Ames2019control}: one uses a nominal controller incorporating a CBF constraint \cite[Equation~(CBF-QP)]{Ames2019control}, the other  utilizes a CLF and a CBF (\cite[Equation~(CLF-CBF-QP)]{Ames2019control}, also in \eqref{eq:original_QP}). In our proposed formulation, both a CLF $V(\myvar{x})$ and a CLF-compatible  $\myvar{u}_{nom}$ are needed. One could view our modification as a combination of the two formulations in \cite{Ames2019control}. The rationale behind this modification is that, given  a CLF $V(\myvar{x})$, calculating $\myvar{u}_{nom}$ is straightforward from \cite{sontag1989universal}, and adding this extra information will guide the selection of a stabilizing control input from all the feasible inputs. Based on this interpretation, the resulting controller being similar to $u_{nom}$ is an intended result. In Theorem \ref{thm:safe_and_stability}, we prove that this modification helps removing the undesired equilibria and aligning the resulting controller to a stabilizing controller.
\end{remark}

\vspace{2mm}
\begin{remark}
It is tempting to claim from Theorem \ref{thm:safe_and_stability} that the resulting controller guarantees that all integral curves converge to origin. Yet  in general  this is not true because 1) the integral curves may converge to the equilibrium points on $\partial \myset{C}$; 2) limit cycles, or other types of attractors may exist in the closed-loop system. We note that the possibility of system trajectory converging to the boundary equilibirum point is not a result of our modification, but an inherent property of the CLF-CBF-QP formulation as discussed in Proposition \ref{prop:boundary_equilibrium}. Actually, for the scenario in Fig. \ref{fig:simulated_example1_trajectory}, global convergence  with a smooth vector field is impossible due to topological obstruction \cite{koditschek1990robot}.
\end{remark}

\vspace{2mm}
\bluetext{
\begin{remark}[Region of attraction]
One can establish a conservative estimate of the region of attraction (ROA)  for the modified CLF-CBF-QP controller. This ROA is given by a sub-level set of the control Lyapunov function, within which the CLF condition holds. Consider the QP in \eqref{eq:QP_with_u_prime}. From Theorem \ref{thm:exact_qp_solution} and the fact that $ F^{\prime}_V(\myvar{x}) \leq 0 $ for all $\myvar{x} \in \mathbb{R}^n$, we know that for all  $\myvar{x}\in \Omega^{\overline{clf}}_{\overline{cbf}}\cup \Omega^{\overline{clf}}_{cbf,1}\cup \Omega^{\overline{clf}}_{cbf,2}\cup \Omega^{clf}_{\overline{cbf}} \cup \Omega^{clf}_{cbf,1}$,  $  L_{\myvarfrak{f}} V(\myvar{x}) + L_{\myvarfrak{g}}V(\myvar{x})\myvar{u} = L_{\myvarfrak{f}^{\prime}} V(\myvar{x}) + L_{\myvarfrak{g}}V(\myvar{x})\myvar{u}^{\prime} \leq - \gamma(V(\myvar{x})) $.  Define $\myset{A}_a:= \{ \myvar{x}: V(\myvar{x}) \leq a \} $ for some $a>0$, and  $\eta:=\arg\sup_{a} a  $ such that $\myset{A}_{a}\cap \Omega^{clf}_{cbf,2} = \emptyset. $ The estimated ROA is then given by $\myset{A}_{\eta}$.
\end{remark}
}

\vspace{2mm}
\begin{example} \label{ex:new_qp_single_integrator}
Consider a mobile robot whose dynamics is given in \eqref{eq:single_integrator_x} with its position  $(x_1,x_2)$ in $\mathbb{R}^2$.  This robot is tasked to navigate to the origin while avoiding a circular region. If the original QP in \eqref{eq:original_QP} is applied, as shown in Fig. \ref{fig:simulated_example3_trajectory},   the mobile robot can at best reach a neighborhood region of the origin, the size of which is determined by $p$. If the new control formulation in \eqref{eq:u_safe_and_stabilizing} is applied, and we choose $\myvar{u}_{nom} = -2x$, then the transformed system dynamics is given in \eqref{eq:single_integrator_-x}. From Fig. \ref{fig:simulated_example1_trajectory}, we observe that the robot can reach the origin, and not merely a neighborhood of it,  no matter what value of $p$ is chosen. We also observed that in both cases, the robot may get stuck at $(0,6)$. \bluetext{In \cite{mestres2022optimization}, a numerical comparison was carried out between our method and the approach therein for this scenario, which results in similar closed-loop system behavior.  We note that our estimated ROA  $ \{\myvar{x}: x_1^2 + x_2^2\leq 3.3^2 \}$ is larger than the counterpart  $\{\myvar{x}: x_1^2 + x_2^2\leq 2^2 \}$ in \cite{mestres2022optimization}.}
\end{example}

\vspace{2mm}
\begin{example} \label{ex:new_qp_high_order}
Now we consider a second-order mobile robot whose dynamics is given in \eqref{eq:high_order_example} with  the position state  $x_1$ and velocity state $x_2$. This robot is tasked to navigate to $0$ while its state needs to avoid the region in dark green in  Fig. \ref{fig:simulated_example2_trajectory}. If the original QP in \eqref{eq:original_QP} is applied, then the robot will stop at certain undesired points instead of the $0$ position. If the new control formulation in \eqref{eq:u_safe_and_stabilizing} is applied, and we choose $\myvar{u}_{nom} = -2x_1 -x_2$, then the transformed system dynamics is $    \begin{psmallmatrix}
    \dot{x}_1 \\
    \dot{x}_2
    \end{psmallmatrix} = \begin{psmallmatrix}
      x_2\\
     -x_1 -x_2
    \end{psmallmatrix} + \begin{psmallmatrix}
    0\\
    1
    \end{psmallmatrix} u^{\prime}.$
With the same CLF and CBF functions as in Example \ref{ex:high-order}, the robot reaches exactly $0$ position, not merely a neighborhood of it or a position on the safety boundary, no matter what value of $p$ is chosen as shown in Fig. \ref{fig:simulated_example4_trajectory}.
\end{example}

There remain several open problems that are worthy of  investigation. These include 1)  control design using hybrid control or time-varying CBF to achieve safety and (almost) global convergence; 2) \bluetext{analyzing closed-loop system behavior with input bounds and multiple compatible CBFs \cite{tan2022compatibility}. } All these problems, however, are out of the scope of this work and require future endeavors.

\section{Conclusion}
 In this paper, we have characterized, for general control-affine systems with a CLF-CBF based quadratic program in the loop, the existence and locations of all possible closed-loop equilibrium points. We further provide analytical results on how the parameter in the program should be chosen to remove the undesired equilibrium points or to confine them in a small neighborhood of the origin. Our main result, a modified quadratic program formulation, is then presented. With the mere assumptions on the existence of a CLF and a CBF, the proposed formulation  guarantees simultaneously the forward invariance of the safety set, the complete elimination of undesired equilibrium points in the interior of it,  the elimination of undesired boundary equilibria with $L_{\myvarfrak{g}}h = \myvar{0}$, and the local asymptotic stability of the origin.


\bibliographystyle{unsrt}        
\bibliography{ref}           

\end{document}